\documentclass[a4paper]{llncs}
\usepackage[utf8]{inputenc}
\usepackage[T1]{fontenc}
\usepackage{amsmath}
\usepackage{amssymb}
\usepackage[dvipsnames]{xcolor}
\usepackage{tikz}
\usetikzlibrary{calc}
\usetikzlibrary{patterns,fadings}
\usetikzlibrary{decorations.pathreplacing}
\usepackage{booktabs}
\usepackage{xspace}
\usepackage{algorithm}
\usepackage[noend]{algorithmic}
\floatname{algorithm}{Algorithm}

\usepackage{url}

\newcommand{\jobAShort}[6]{%
 \begin{scope}[yshift=#1]
  \draw[draw=black,fill=red!30] (#2,0) rectangle ($(#2,0)+(#3,0.5)$);%
  \draw[draw=black,fill=red!30] (#4,0) rectangle ($(#4,0)+(#3,0.5)$);%
  \node (joblabel) at ($(#2,0)+(0,0.25)+0.5*(#3,0)$) {\footnotesize \ensuremath{#5}};%
  \node (joblabel) at ($(#4,0)+(0,0.25)+0.5*(#3,0)$) {\footnotesize \ensuremath{#6}};%
 \end{scope}
}
\newcommand{\jobBShort}[7]{%
 \begin{scope}[yshift=#1]
  \draw[draw=black,fill=green!30] (#2,0) rectangle ($(#2,0)+(#3,0.5)$);%
  \node (a) at ($(#2,0)+(0,0.25)+0.5*(#3,0)$) {\footnotesize \ensuremath{#4}};%
  \draw[thick] ($(#6,0)+(0,0.1)$) -- ($(#6,0)+(0,0.4)$);
  \draw[thick] ($(#7,0)+(0,0.1)$) -- ($(#7,0)+(0,0.4)$);
  \draw[thick] ($(#6,0)+(0,0.25)$) -- ($(#2,0)+(0,0.25)$);
  \draw[thick] ($(#7,0)+(0,0.25)$) -- ($(#2,0)+(#3,0.25)$);
 \end{scope}
}
\newcommand{\jobCShort}[7]{%
  \draw[draw=black,fill=blue!30] (#2,#1) rectangle ($(#2,#1)+(#3+1,0.5)$);%
 \node (a) at ($(#2,#1)+(0.5,0.25)+0.5*(#3,0)$) {\footnotesize \ensuremath{#4}};%
 \draw[thick] ($(#6,#1)+(0,0.25)$) -- ($(#2,#1)+(0,0.25)$);
 \draw[thick] ($(#7,#1)+(0,0.25)$) -- ($(#2,#1)+(#3,0.25)+(1,0)$);
}

\newcommand{\introduceproblem}[3]{%
\begin{quote}
#1\\
\textbf{Input:} #2\\
\textbf{Problem:} #3
\end{quote}}

\newcommand{\N}{\mathbb{N}}

\newcommand{\ama}{\[ \begin{aligned}}
\newcommand{\ema}{\end{aligned} \]}
\newcommand{\set}[1]{\left\{#1\right\}}
\newcommand{\setc}[2]{\left\{\left.#1\ \right|\ #2\right\}}
\renewcommand{\P}{\ensuremath{\mathsf{P}}\xspace}
\newcommand{\NP}{\ensuremath{\mathsf{NP}}\xspace}
\newcommand{\PARTITION}{\ensuremath{\mathsf{PARTITION}}\xspace}
\newcommand{\MAXCONNECTIVITY}{\ensuremath{\mathsf{MAXCONNECTIVITY}}\xspace}
\newcommand{\MINCONNECTIVITY}{\ensuremath{\mathsf{MINCONNECTIVITY}}\xspace}
\newcommand{\THREESAT}{\ensuremath{\mathsf{3SAT}}\xspace}
\newcommand{\NO}{\ensuremath{\mathsf{NO}}\xspace}
\newcommand{\YES}{\ensuremath{\mathsf{YES}}\xspace}
\newcommand{\SATFALSE}{\ensuremath{\mathsf{FALSE}}\xspace}
\newcommand{\SATTRUE}{\ensuremath{\mathsf{TRUE}}\xspace}

\title{Scheduling Maintenance Jobs in Networks\thanks{This work is supported by the German Research Foundation (DFG) under project ME 3825/1 and within project A07 of CRC TRR 154. It is partially funded in the framework of \textsc{Matheon} supported by the Einstein Foundation Berlin and by the Alexander von Humboldt Foundation.}}
\titlerunning{Scheduling Maintenance Jobs in Networks}
\author{Fidaa Abed\inst{1} \and Lin Chen\inst{2} \and Yann Disser\inst{3} \and Martin Gro\ss{}\inst{4} \and Nicole Megow\inst{5} \and Julie Mei\ss{}ner\inst{6} \and Alexander T. Richter\inst{7} \and Roman Rischke\inst{8}} 
\institute{University of Jeddah, Jeddah, Saudi Arabia.\\\email{fabed@uj.edu.sa}
\and University of Houston, Texas, USA.\\\email{chenlin198662@gmail.com}%
\and TU Darmstadt, Darmstadt, Germany.\\\email{disser@mathematik.tu-darmstadt.de}%
\and University of Waterloo, Waterloo, Ontario, Canada.\\\email{mgrob@uwaterloo.ca}%
\and University of Bremen, Bremen, Germany.\\\email{nicole.megow@uni-bremen.de}%
\and TU Berlin, Berlin, Germany.\\\email{jmeiss@math.tu-berlin.de}%
\and TU Braunschweig, Braunschweig, Germany.\\\email{a.richter@tu-bs.de}%
\and TU M{\"u}nchen, M{\"u}nchen, Germany.\\\email{rischke@ma.tum.de}}	

\begin{document}

\maketitle

\begin{abstract}
We investigate the problem of scheduling the maintenance of edges in a network, motivated by the goal of minimizing outages in transportation or telecommunication networks. We focus on maintaining connectivity between two nodes over time; for the special case of path networks, this is related to the problem of minimizing the busy time of machines.

We show that the problem can be solved in polynomial time in arbitrary networks if preemption is allowed. If preemption is restricted to integral time points, the problem is NP-hard and in the non-preemptive case we give strong non-approximability results. Furthermore, we give tight bounds on the power of preemption, that is, the maximum ratio of the values of non-preemptive and preemptive optimal solutions.

Interestingly, the preemptive and the non-preemptive problem can be solved efficiently on paths, whereas we show that mixing both leads to a weakly NP-hard problem that allows for a simple 2-approximation.

\end{abstract}

\section{Introduction} 
Transportation and telecommunication networks are important backbones of modern infra\-structure and have been a major focus of research in combinatorial optimization and other areas. 
Research on such networks usually concentrates on optimizing their usage, for example by maximizing throughput or minimizing costs. 
In the majority of the studied optimization models it is assumed that the network is permanently available, and our choices only consist in deciding which parts of the network to use at each point in time.

Practical transportation and telecommunication networks, however, can generally not be used non-stop. 
Be it due to wear-and-tear, repairs, or modernizations of the network, there are times when parts of the network are unavailable. 
We study how to schedule and coordinate such maintenance in different parts of the network to ensure connectivity.

While network problems and scheduling problems individually are fairly well understood, the combination of both areas that results from scheduling network maintenance has only recently received some attention~\cite{BolandKK15,BolandKWZ14,NurreEtAl12,Bley13,FlamminiEtAl10} and is theoretically hardly understood.

\medskip
\noindent\textbf{Problem Definition.} 
In this paper, we study connectivity problems which are fundamental in this context.
In these problems, we aim to schedule the maintenance of edges in a network in such a way as to preserve connectivity between two designated vertices. Given a network and maintenance jobs with processing
times and feasible time windows, we need to decide on the temporal
allocation of  the maintenance jobs. While a maintenance on an edge is performed, the edge is
not available. We distinguish between \MINCONNECTIVITY, the problem in
which we minimize the total time in which the network is disconnected,
and \MAXCONNECTIVITY, the problem in which we maximize the total time in which it is connected.

In both of these problems, we are given an undirected graph $G=(V,E)$ with two distinguished vertices $s^+,s^- \in V$. 
We assume w.\,l.\,o.\,g.\, that the graph is simple; we can
replace a parallel edge $\set{u,w}$ by a new node $v$ and two edges
$\set{u,v},\set{v,w}$.   
Every edge $e\in E$ needs to undergo $p_e \in \mathbb{Z}_{\geq 0}$
time units of maintenance within the time window $[r_e,d_e]$ with
$r_e,d_e \in \mathbb{Z}_{\geq 0}$, where $r_e$ is called the release
date and $d_e$ is called the deadline of the maintenance job for
edge~$e$.  An edge $e=\{u,v\}\in E$ that is maintained at time $t$, is not available at $t$ in the graph $G$. 
We consider preemptive and non-preemptive maintenance jobs. If a job
must be scheduled non-preemptively then, once it is started, it must
run until completion without any interruption. If a job is allowed to
be preempted, then its processing can be interrupted at any time and
may resume at any later time without incurring extra cost.

A \emph{schedule} $S$ for $G$ assigns the maintenance job of every edge $e \in E$ to a single time interval (if non-preemptive) or a set of disjoint time intervals (if preemptive) $S(e) :=  \{[a_1,b_1], \ldots, [a_k,b_k]\}$ with
$$r_e \leq a_i \leq b_i \leq d_e, \text{ for $i \in [k]$ and } \sum_{[a,b] \in S(e)} (b-a) = p_e .$$ 
If not specified differently, we define $T:=\max_{e\in E} d_e$ as our \emph{time horizon}. 
We do not limit the number of simultaneously maintained edges.

For a given maintenance schedule, we say that the network $G$ is \emph{disconnected at time~$t$} if there is no path from $s^+$ to $s^-$ in $G$ at time $t$, otherwise we call the network $G$ \emph{connected at time~$t$}. 
The goal is to find a maintenance schedule for the network $G$ so that the total time where~$G$ is disconnected is minimized (\MINCONNECTIVITY).
We also study the maximization variant of the problem, in which we want to find a schedule that maximizes the total time where $G$ is connected (\MAXCONNECTIVITY).

\medskip
\noindent\textbf{Our Results.}
For \emph{preemptive} maintenance jobs, we show that we can solve both
problems, \MAXCONNECTIVITY and \MINCONNECTIVITY,
efficiently in arbitrary networks~(Theorem~\ref{thm:pmtn}). This result crucially requires that
we are free to preempt jobs at arbitrary points in time. 
Under the restriction that we can {\em preempt} jobs only at {\em integral points
in time}, the problem becomes \NP-hard. More specifically, \MAXCONNECTIVITY does not admit a
$(2-\epsilon)$-approximation algorithm for any $\epsilon > 0$ in this
case, and \MINCONNECTIVITY is inapproximable
(Theorem~\ref{thm:integralPmtn}), unless
$\P=\NP$. By inapproximable, we mean that it is \NP-complete to decide
whether the optimal objective value is zero or positive, leading to
unbounded approximation factors. 
                               
This is true even for unit-size
jobs. This complexity result is
interesting and may be surprising, as it is in contrast to results for
standard scheduling problems, without an underlying network. 
Here, the restriction to integral preemption typically
does not increase the problem complexity when all other input 
parameters are integral. 
However, the same question remains open in a related problem concerning the busy-time in scheduling, studied in~\cite{khuller-mapsp,ChangKhullerMukherjee14}.

For \emph{non-preemptive} instances, we establish that there is no
$(c\sqrt[3]{|E|})$-ap\-prox\-ima\-tion algorithm for \MAXCONNECTIVITY
for some constant $c>0$ and that \MINCONNECTIVITY
is inapproximable even on disjoint paths between two nodes $s$ and $t$, unless $\P = \NP$ (Theorems~\ref{thm:no-pmtnLowerBound1},\ref{thm:no-pmtnLowerBound2}). 
On the positive side, we provide an
$(\ell+1)$-approximation algorithm for \MAXCONNECTIVITY in
general graphs~(Theorem~\ref{thm:napprox}), where $\ell$ is the number of distinct latest start times (deadline minus processing time) for jobs.

We use the notion \emph{power of preemption} to capture the benefit
of allowing arbitrary job preemption. The power of preemption is a commonly used measure for the impact of
preemption in scheduling~\cite{CanettiIrani98,CorreaSkutellaVerschae12,SchulzSkut02,SoperS14}. Other
terms used in this context include \emph{price of
  non-preemption}~\cite{CohenAddadEtAl15}, \emph{benefit of preemption}~\cite{ParsonsSevcik95} and \emph{gain of
  preemption}~\cite{Ha92}. It is defined as the maximum ratio
of the objective values of an optimal non-preemptive and an optimal preemptive
solution. We show that the power of preemption is $\Theta(\log |E|)$ for  \MINCONNECTIVITY on a path (Theorem~\ref{thm:powerofpreemption}) and unbounded for  \MAXCONNECTIVITY on a path
 (Theorem~\ref{thm:powerofpreemptionmax}). This is in contrast to
 other scheduling problems, where the power of
 preemption is constant,
 e.\,g.\! \cite{CorreaSkutellaVerschae12,SchulzSkut02}.
 
On paths, we show that \emph{mixed} instances, which have both preemptive and non-preemptive jobs, are weakly \NP-hard
~(Theorem~\ref{theorem:mixed}). This hardness result is of particular interest, 
as both purely non-preemptive and purely preemptive instances can be
solved efficiently on a path~(see Theorem~\ref{thm:pmtn} and
\cite{KhandekarEtAl15}). Furthermore, we give a simple $2$-approximation algorithm
for mixed instances of
\MINCONNECTIVITY~(Theorem~\ref{theorem:mixedapprox}).  

\medskip
\noindent\textbf{Related Work.} 
The concept of combining scheduling with network problems has been considered by different communities lately. However, the specific problem of only maintaining connectivity over time between two designated nodes has not been studied to our knowledge. Boland et al.~\cite{BolandKK15,BolandNKK15,BolandKWZ14} study the combination of non-preemptive arc maintenance in a transport network, motivated by annual maintenance planning for the Hunter Valley Coal Chain~\cite{BolandS12}. Their goal is to schedule maintenance such that the maximum $s$-$t$-flow over time in the network with zero transit times is maximized. They show strong \NP-hardness for their problem and describe various heuristics and IP based methods to address it. Also, they show in~\cite{BolandNKK15} that in their non-preemptive setting, if the input is integer, there is always an optimal solution that starts all jobs at integer time points. In~\cite{BolandKK15}, they consider a variant of their problem,
where the number of concurrently performable maintenances is bounded
by a constant. 

Their model generalizes ours in two ways -- it has capacities and the objective is to maximize the total flow value. As a consequence of this, their IP-based methods carry over to our setting, but these methods are of course not efficient. Their hardness results do not carry over, since they rely on the capacities and the different objective. However, our hardness results -- in particular our approximation hardness results -- carry over to their setting, illustrating why their IP-based models are a good approach for some of these problems.

Bley, Karch and D'Andreagiovanni~\cite{Bley13} study how to upgrade a telecommunication network to
a new technology employing a bounded number of technicians. Their goal is to minimize the total service disruption caused by downtimes. A major
difference to our problem is that there is a set of given paths that shall be upgraded and a path can only be used if it is either completely upgraded or not upgraded. They give ILP-based approaches for
solving this problem and show strong \NP-hardness for a non-constant number of paths by reduction from the linear arrangement
problem. 

Nurre et al.~\cite{NurreEtAl12} consider the problem of
restoring arcs in a network after a major disruption, with restoration
per time step being bounded by the available work force. Such network design problems over time 
have also been considered by Kalinowski, Matsypura and Savelsbergh~\cite{KalinowskiMatsypuraSavelsbergh15}.  

In scheduling, minimizing the busy time refers to minimizing the
amount of time for which a machine is used. Such problems have
applications for instance in the context of energy management~\cite{MertziosEtAl12} or fiber management in optical
networks~\cite{FlamminiEtAl10}. They have been studied from the
complexity and approximation point of view in~\cite{ChangKhullerMukherjee14,FlamminiEtAl10,KhandekarEtAl15,MertziosEtAl12}. 
The problem of minimizing the busy time is equivalent to our problem in the case of a path, because there we have connectivity at a time point when no edge in the path is maintained, i.\,e., no machine is busy.

Thus, the results of Khandekar et al.~\cite{KhandekarEtAl15} and Chang, Khuller and Mukherjee~\cite{ChangKhullerMukherjee14} have direct implications for us.
They show that minimizing busy time can be done efficiently for purely non-preemptive and purely preemptive instances, respectively. 
\section{Preemptive Scheduling}\label{sec:preemptive}
In this section, we consider problem instances where all maintenance jobs can be preempted.

\begin{theorem}\label{thm:pmtn}
  Both \MAXCONNECTIVITY and \MINCONNECTIVITY with pre\-emp\-tive jobs
  can  be solved optimally in polynomial time on arbitrary graphs.	
\end{theorem}

\begin{proof}
 We establish a linear program (LP) for \MAXCONNECTIVITY. 
 Let $TP = \{0\}\cup\{r_e,d_e: e \in E\} = \set{t_0,t_1,\dots,t_{k}}$ be the set of all \emph{relevant time points} with $t_0 < t_1 < \dots < t_k$. We define $I_i := [t_{i-1},t_i]$ and $w_i := |I_i|$ to be the length of interval $I_i$ for $i=1,\dots,k$. 
 
 In our linear program we model connectivity during interval $I_i$ by an $(s^+,s^-)$-flow $x^{(i)}$, $i \in \{1,\ldots,k\}$.
 To do so, we add for every undirected edge $e=\{u,v\}$ two directed arcs $(u,v)$ and $(v,u)$. 
 Let $A$ be the resulting arc set.
 With each edge/arc we associate a capacity variable $y^{(i)}_{e}$, which represents the fraction of availability of edge $e$ in interval $I_i$. Hence, $1-y^{(i)}_{e}$ gives the relative amount of time spent on the maintenance of edge $e$ in $I_i$. 
 Additionally, the variable~$f_i$ expresses the fraction of availability for interval $I_i$.
 \begin{align} 
\max &&\sum_{i=1}^k w_i \cdot f_i \\
\text{s.t.} && \sum_{u:(v,u) \in A}\! x^{(i)}_{(v,u)} - \sum_{u:(u,v) \in A}\! x^{(i)}_{(u,v)}& = \begin{cases}
                  f_i        & \forall\,i\in[k],\, v=s^+, \\
                    0        & \forall\,i\in[k],\, v\in V \setminus \{s^+,s^-\},\label{ineq:flowcons}\\
                     -f_i		& \forall\,i\in[k],\, v=s^-,
                   \end{cases}\\
&& \sum_{i:I_i \subseteq [r_e,d_e]} (1-y^{(i)}_{e})w_i & \geq p_e
&&\hspace{-4cm}\forall\, e\in E,
 \label{ineq:processing}
\\ 
&&  x^{(i)}_{(u,v)},x^{(i)}_{(v,u)} & \leq y^{(i)}_{\{u,v\}}
&& \hspace{-4cm}\forall\, i \in [k],\, \{u,v\}\in E,
 \label{ineq:bound1}
\\
&& f_i&\leq 1 
&&\hspace{-4cm}\forall\, i \in [k],
&& \label{ineq:sourceflow}
\\
&& x^{(i)}_{(u,v)},x^{(i)}_{(v,u)},y^{(i)}_{\{u,v\}} &\in [0,1]
&&\hspace{-4cm}\forall\, i \in [k],\, \{u,v\} \in E\label{ineq:nonneg}.
\end{align}  
  Notice that the LP is polynomial in the input size, since $k \leq 2|E|$. We show in Lemma~\ref{Lem:PreemptiveRelaxation} that this LP is a relaxation of preemptive \MAXCONNECTIVITY on general graphs and in Lemma~\ref{Lem:PreemptiveTrans} that any optimal solution to it can be turned into a feasible schedule with the same objective function value in polynomial time, which proves the claim for \MAXCONNECTIVITY. For \MINCONNECTIVITY, notice that any solution that maximizes the time in which $s$ and $t$ are connected also minimizes the time in which $s$ and $t$ are disconnected -- thus, we can use the above LP there as well.
\qed
\end{proof}

Next, we need to prove the two lemmas that we used in the proof of Theorem~\ref{thm:pmtn}. We begin by showing that the LP is indeed a relaxation of our problem.

 \begin{lemma}
  \label{Lem:PreemptiveRelaxation}
  The given LP is a relaxation of preemptive \MAXCONNECTIVITY on general graphs.	
 \end{lemma}
 
\begin{proof}
Given a feasible maintenance schedule, consider an arbitrary interval $I_i$, $i\in \{1,\ldots,k\}$, and let $[a^i_1,b^i_1]\,\dot{\cup}\,\ldots \dot{\cup}\,[a^i_{m_i},b^i_{m_i}] \subseteq I_i$ be all intervals where $s^+$ and $s^-$ are connected in interval~$I_i$.
We set $f_i= \sum_{\ell=1}^{m_i} (b^i_{\ell} -a^i_{\ell})/ w_i \leq 1$ and set $y^{(i)}_{e} \in [0,1]$ to the fraction of time where edge $e$ is not maintained in interval $I_i$.
Note that \eqref{ineq:processing} is automatically fulfilled, since we consider a feasible schedule. 
It is left to construct a feasible flow $x^{(i)}$ for the fixed variables $f_i$ and $y^{(i)}$ for all $i=1,\ldots,k$.

Whenever the given schedule admits connectivity we can send one unit of flow from $s^+$ to $s^-$ along some directed path in $G$.
Moreover, in intervals where the set of processed edges does not change we can use the same path for sending the flow.
Let $[a,b]\subseteq I_i$ be an interval where the set of processed edges does not change and in which we have connectivity. Let $\mathcal{C}_i$ be the collection of all such intervals in $I_i$.
Then, we send a flow $x^{(i)}_{[a,b]}$ from $s^+$ to $s^-$ along any path of total value $(b-a)/w_i$ using only arcs for which the corresponding edge is not processed in $[a,b]$.
The flow $x^{(i)} = \sum_{[a,b]\in \mathcal{C}_i} x^{(i)}_{[a,b]}$, which is a sum of vectors, gives the desired flow. 
The constructed flow $x^{(i)}$ respects the flow conservation \eqref{ineq:flowcons} and non-negativity constraints \eqref{ineq:nonneg}, uses no arc more than the corresponding $y^{(i)}_{e}$, since flow $x^{(i)}$ is driven by the schedule. \qed
\end{proof}

 \begin{lemma}
  \label{Lem:PreemptiveTrans}
  Any feasible LP solution can be turned into a feasible maintenance schedule at no loss in the objective function value in polynomial time.	
 \end{lemma}
\begin{proof} 
Let $(x,y,f)$ be a feasible solution of the given LP. Let $\mathcal{P}^i:=(P^i_1,\ldots,P^i_{\lambda_i})$ be a path decomposition~\cite{PathDecomp} of the $(s^+,s^-)$-flow~$x^{(i)}$ for an arbitrary interval $I_i:=[a_i,b_i]$, $i\in \{1,\ldots,k\}$, after deleting all flow from possible circulations.
Furthermore, let $x(P^i_\ell)$ be the value of the $(s^+,s^-)$-flow $x^{(i)}$ sent along the directed path~$P^i_\ell$.
For each arc $a \in A$ we have that $\sum_{\ell \in[\lambda_i]: a \in P^i_\ell} x(P^i_\ell) = x^{(i)}_{a}$ by the definition of $\mathcal{P}^i$.
Hence, we get $\sum_{\ell \in [\lambda_i]} x(P^i_\ell) = f_i \leq 1$ by using \eqref{ineq:sourceflow}. 
We now divide the interval~$I_i$ into disjoint subintervals to allocate connectivity time for each path in our path decomposition.
More precisely, we do \emph{not} maintain any arc $(u,v)$ (resp. edge $\{u,v\}$) contained in  $P^i_\ell$, $\ell = 1,\ldots,\lambda_i$, in the time interval
\begin{equation}
 \left[a_i+\sum_{m=1}^{\ell -1} w_i \cdot x(P^i_m), a_i+ \sum_{m =1}^{\ell} w_i \cdot x(P^i_m)\right] \text{ of length } w_i \cdot x(P^i_{\ell}) .
\end{equation}
Inequality~\eqref{ineq:bound1} and $\sum_{\ell \in[\lambda_i]: a \in P^i_\ell} x(P^i_\ell) = x^{(i)}_{a}$ thereby ensure that by now the total time where edge~$e$ does not undergo maintenance in interval $I_i$ equals at most $w_i \cdot y^{(i)}_{e}$ time units. 
By Inequality~\eqref{ineq:processing}, we can thus distribute the processing time of the job for edge $e$ among the remaining slots of all intervals~$I_i$, $i =1,\ldots,k$.
For instance, we could greedily process the job for edge~$e$ as early as possible in available intervals.
Note that arbitrary preemption of the processing is allowed.
By construction, we have connectivity on path  $P^i_\ell$, $\ell =1,\ldots, \lambda_i$, for at least $w_i \cdot x(P^i_\ell)$ time units in interval $I_i$. 
Thus, the constructed schedule has total connectivity time of at least $\sum_{i=1}^k w_i \sum_{\ell=1}^{\lambda_i} x(P^i_\ell) = \sum_{i=1}^k w_i \cdot f_i$. Since the path decomposition can be computed in polynomial-time and the resulting number of paths is bounded by the number of edges~\cite{PathDecomp}, we can obtain the feasible schedule in polynomial-time. \qed
\end{proof}

For unit-size jobs we can simplify the given LP by restricting to the first $|E|$ slots within every interval $I_i$. 
This, in turn, allows to consider intervals of unit-size, i.e., we have $w_i=1$ for all intervals $I_i$, which affects constraint~\eqref{ineq:processing}.
However, one can show that the constraint matrix of this LP is generally not totally unimodular.
We illustrate the behaviour of the LP with the help of the following exemplary instance in Figure~\ref{fig:pmtnIntegral}, in which all edges have unit-size jobs associated and the label of an edge $e$ represents $(r_e,d_e)$.
It is easy to verify that a schedule that preempts jobs only at integral time points, has maximum connectivity time of one.
However, the following schedule with arbitrary preemption has connectivity time of two.
We process $\{s^+,v_2\}$ in $[0,0.5] \cup [1,1.5]$, $\{s^+,v_3\}$ in $[0.5,1] \cup [1.5,2]$, $\{v_4,s^-\}$ in $[0,0.5] \cup [1.5,2]$, $\{v_5,s^-\}$ in $[0.5,1.5]$, and the other edges are fixed by their time window.
This instance shows that the integrality gap of the LP is at least two.
 
\begin{figure}[htb]
  \centering
  \tikzstyle{vertex}=[draw,circle,inner sep=1pt,minimum size=22pt,scale=1.0]
  \tikzstyle{edge} = [draw,thick,-]
  \tikzstyle{weight} = [sloped,above,font=\small]
  \begin{tikzpicture}[scale=1.0]
    \node[vertex] (v1) at (0,0) {$s^+$};
    \node[vertex] (v2) at (2,2) {$v_2$};
    \node[vertex] (v3) at (2,-2) {$v_3$};
    \node[vertex] (v4) at (5,2) {$v_4$};
    \node[vertex] (v5) at (5,-2) {$v_5$};
    \node[vertex] (v6) at (7,0) {$s^-$};
    \path[edge] (v1) -- node[weight] {$(0,2)$} (v2);
    \path[edge] (v1) -- node[weight] {$(0,2)$} (v3);
    \path[edge] (v2) -- node[weight] {$(1,2)$} (v4);
    \path[edge] (v2) -- node[weight,near start] {$(0,1)$} (v5);
    \path[edge] (v3) -- node[weight,near start] {$(0,1)$} (v4);
    \path[edge] (v3) -- node[weight] {$(1,2)$} (v5);
    \path[edge] (v4) -- node[weight] {$(0,2)$} (v6);
    \path[edge] (v5) -- node[weight] {$(0,2)$} (v6);
  \end{tikzpicture}
  \caption{Example for the difference between arbitrary preemption and preemption only at integral time points.}\label{fig:pmtnIntegral}
\end{figure}
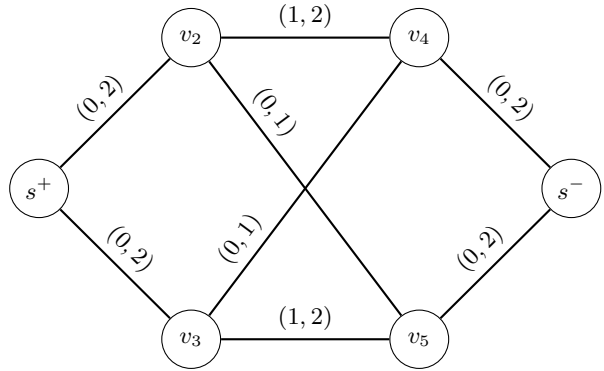

The statement of Theorem~\ref{thm:pmtn} crucially relies on the fact that we may preempt jobs arbitrarily. 
However, if preemption is only possible at integral time points, the problem becomes \NP-hard even for unit-size jobs. This follows from the proof of Theorem~\ref{thm:no-pmtnLowerBound1} for $t_1=0$, $t_2=1$, and $T=2$.

\begin{theorem}\label{thm:integralPmtn}
 \MAXCONNECTIVITY with preemption only at integral time points is NP-hard and does not admit a $(2-\epsilon)$-approximation algorithm for any $\epsilon > 0$, unless $\P=\NP$.
	Furthermore, \MINCONNECTIVITY with preemption only at integral time points is inapproximable.
\end{theorem}
\section{Non-Preemptive Scheduling}\label{sec:nonpreemptive_hardness} 
We consider problem instances in which no 
job can be
preempted. We show that there is no $(c\sqrt[3]{|E|})$-approximation
algorithm for
\MAXCONNECTIVITY for some $c > 0$. We also show that \MINCONNECTIVITY is inapproximable, unless $\P = \NP$. Furthermore, we give an $(\ell+1)$-approximation algorithm, where $\ell := |\setc{d_e-p_e}{e\in E}|$ is the number of distinct latest start times for jobs. 

To show the strong hardness of approximation for
\MAXCONNECTIVITY, we begin with a weaker result which provides us with a
crucial gadget. 

\begin{theorem}
 \label{thm:no-pmtnLowerBound1}
 Non-preemptive \MAXCONNECTIVITY does not admit a
 $(2-\epsilon)$-approximation algorithm, for $\epsilon > 0$, and
 non-preemptive \MINCONNECTIVITY is inapproximable, unless $\P=\NP$. This holds even for unit-size jobs.
\end{theorem}

\begin{proof}
  We show that the existence of a $(2-\epsilon)$-approximation algorithm for non-preemptive \MAXCONNECTIVITY allows to distinguish between \YES- and \NO-instances of \THREESAT in polynomial time. 
  Given an instance of \THREESAT consisting of $m$ clauses $C_1,C_2,\dots,C_m$ each of exactly three variables in $X=\{x_1, x_2, \dots, x_n\}$, we construct the following instance of non-preemptive \MAXCONNECTIVITY.
  We pick two arbitrary but distinct time points $t_1+1 \leq t_2$ and a polynomially bounded time horizon $T \geq t_2+1$. 
  We construct our instance such that connectivity is impossible outside $[t_1,t_1+1]$ and $[t_2,t_2+1]$.
  For this, $s^+$ is followed by a path $P$ from $s^+$ to a vertex $s'$ composed of three edges that disconnect $s^+$ from $s^-$ in the time intervals $[0,t_1]$, $[t_1+1,t_2]$, and $[t_2+1,T]$.
  These edges $e$ have $p_e = d_e - r_e$.
  Furthermore, we construct the network such that the total connectivity time is greater than one if and only if the \THREESAT-instance is a \YES-instance. And we show that if the total connectivity time is greater than one, then there is a schedule with maximum total connectivity time of two. The high-level structure of the graph we will be creating can be found in Figure~\ref{figure:connectivity_variablesS}, with an expanded version following later in Figure~\ref{figure:connectivity_variables}.
	
 \begin{figure}[bht]
   \centering
  \includegraphics[width=\textwidth]{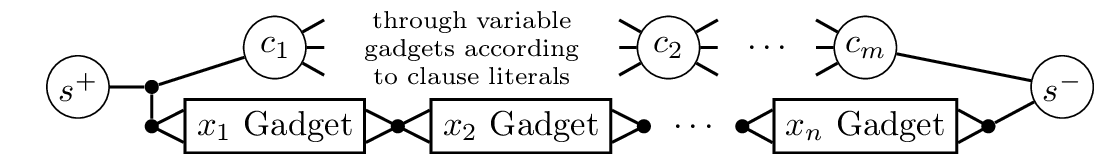}
  \caption{High-level view of the construction for Theorem~\ref{thm:no-pmtnLowerBound1}. \label{figure:connectivity_variablesS}}
 \end{figure}
  
  Let $Y(x_i)$ be the set of clauses containing the literal $x_i$ and $Z(x_i)$ be the set of clauses containing the literal~$\neg x_i$, and set set $k_i = 2 |Y(x_i)|$ and $\ell_i= 2|Z(x_i)|$. 
  We define the following node sets 
  \begin{itemize} 
   \item $V_1:= \{ y_i^1,\ldots,y_i^{k_i}\mid i=1,\ldots,n\}$,
   \item $V_2:=\{ z_i^1,\ldots,z_i^{\ell_i} \mid i=1,\ldots,n\}$, 
   \item $V_3:=\{ c_r \mid r=1,\ldots,m+1\}$,
   \item $V_4:=\{v_i \mid i=1,\ldots,n+1\}$ 
   \item and set $V= \bigcup_{j=1}^4 V_j \cup \{v : v \in P\} \cup \{s^-\}$.   
  \end{itemize}
  We introduce three edge types
  \begin{itemize}
   \item $\mathcal{E}_1:=\{e\in E: r_e=t_1, d_e=t_2+1, p_e=t_2-t_1 \}$, 
   \item $\mathcal{E}_2:=\{e\in E: r_e=t_1, d_e=t_1+1, p_e=1\}$,
   \item and $\mathcal{E}_3:=\{e\in E: r_e=t_2, d_e=t_2+1, p_e=1 \}$.
  \end{itemize}
  
  The graph $G=(V,E)$ consists of variable gadgets, shown in Figure~\ref{figure:connectivity_variables}, to which we connect the clause nodes $c_r$, $r =1,\ldots,m+1$. 
  We define the following edge sets for the variable gadgets, namely, 
  \begin{itemize}
   \item $E_1:= \{ \{s',v_1\},\{v_{n+1},s^-\} \}$ of type $\mathcal{E}_2$,
   \item $E_2 := \{ \{v_i, y_i^1\}, \{v_i, z_i^1\}, \{y_i^{k_i},v_{i+1}\},\{z_i^{\ell_i},v_{i+1}\}: i=1,\ldots,n\}$ of type $\mathcal{E}_2$, 
   \item $E_3:= \{ \{y_i^q, y_i^{q+1}\} : i=1,\ldots,n; q=1,3,\ldots,k_i-3, k_i-1\}$ of type $\mathcal{E}_1$,
   \item $E_4:= \{ \{z_i^q, z_i^{q+1}\} : i=1,\ldots,n; q=1,3,\ldots,\ell_i-3, \ell_i-1\}$ of type $\mathcal{E}_1$, 
   \item $E_5 := \{ \{y_i^q, y_i^{q+1}\} : i=1,\ldots,n; q=2,4,\ldots,k_i-4, k_i-2\}$ of type $\mathcal{E}_2$,
   \item and $E_6 := \{ \{z_i^q, z_i^{q+1}\} : i=1,\ldots,n; q=2,4,\ldots,\ell_i-4, \ell_i-2\}$ of type $\mathcal{E}_2$.
  \end{itemize}  
  Notice that a variable $x_i$ may only appear positive ($\ell_i=0$) or only negative ($k_i =0$) in our set of clauses. 
  In this case, we also have an edge of type $\mathcal{E}_2$ connecting $v_i$ and $v_{i+1}$ besides the construction for the negative ($z$ nodes) or positive part ($y$ nodes).
  Finally, we add edges to connect the clause nodes to the graph.
  If some positive literal $x_i$ appears in clause $C_r$ and $C_r$ is the $q$-th clause with positive $x_i$, we add the edges $\{c_r,y_i^{2q-1}\}$ and $\{y_i^{2q},c_{r+1}\}$ both of type $\mathcal{E}_3$. 
  Conversely, if some $x_i$ appears negated in $C_r$ and $C_r$ is the $q$-th clause with $\neg  x_i$, we add the edges $\{c_r,z_i^{2q-1}\}$ and $\{z_i^{2q},c_{r+1}\}$ both of type $\mathcal{E}_3$. 
  We also connect $c_1$ and $c_{m+1}$  to the graph by adding $\{s',c_1\}$ and $\{c_{m+1},s^-\}$ of type $\mathcal{E}_3$.
  We define $E$ to be the union of all introduced edges.
  Observe that the network $G$ has $O(n+m)$ nodes and edges.
	
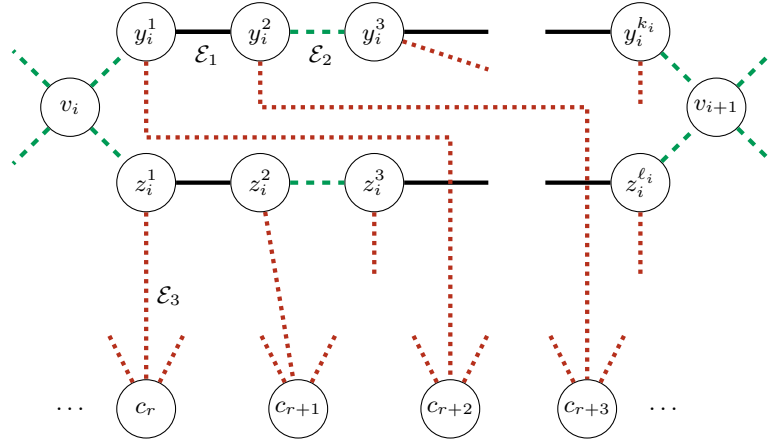
\begin{figure}[htb]
  \centering
  \tikzstyle{switch}=[ultra thick]
  \tikzstyle{zeroblocked}=[switch,dotted,BrickRed]
  \tikzstyle{oneblocked}=[switch,dashed,ForestGreen]
  \begin{tikzpicture}[every node/.style=draw,circle,inner sep=1pt,minimum size=22pt]
   \node (vi)  at (0.00,0.00) {$v_i$};    
   \node (yi1)  at (1.00,1.00) {$y^1_i$};
   \node (zi1)  at (1.00,-1.00) {$z^1_i$};
   \node (yi2)  at (2.50,1.00) {$y^2_i$};
   \node (zi2)  at (2.50,-1.00) {$z^2_i$};   
   \node (yi3)  at (4.00,1.00) {$y^3_i$};
   \node (zi3)  at (4.00,-1.00) {$z^3_i$};      
   \node (yi4)  at (7.50,1.00) {$y^{k_i}_i$};
   \node (zi4)  at (7.50,-1.00) {$z^{\ell_i}_i$};         
   \node (vi1) at (8.50,0.00) {$v_{i+1}$}; 
   \node (cr1) at (1,-4) {$c_r$};
   \node (cr2) at (3,-4) {$c_{r+1}$};
   \node (cr3) at (5,-4) {$c_{r+2}$};
   \node (cr4) at (6.8,-4) {$c_{r+3}$};
   \draw[oneblocked] (vi) -- (yi1);
   \draw[switch] (yi1) -- (yi2);
   \draw[switch] (yi1) -- (yi2);
   \draw[oneblocked] (yi2) -- (yi3);
   \draw[switch,-,path fading=east] (yi3) -- +(1.5,0);
   \draw[switch,-,path fading=west] (yi4) -- +(-1.25,0);
   \draw[oneblocked] (yi4) -- (vi1);
   \draw[oneblocked] (vi) -- (zi1);
   \draw[switch] (zi1) -- (zi2);
   \draw[oneblocked] (zi2) -- (zi3);   
   \draw[switch,-,path fading=east] (zi3) -- +( 1.5,0);
   \draw[switch,-,path fading=west] (zi4) -- +(-1.25,0);
   \draw[oneblocked] (zi4) -- (vi1);
   \draw[zeroblocked,-] (yi1) -- +(0.00,-1.4) -- (5,-0.4) -- (cr3);
   \draw[zeroblocked,-] (yi2) -- +(0.00,-1) -- (6.8,0) -- (cr4);
   \draw[zeroblocked,-,path fading=south] (yi3) -- +(1.5,-0.5);
   \draw[zeroblocked,-,path fading=south] (yi4) -- +(0.00,-1);
   \draw[zeroblocked,-] (zi1) -- (cr1);
   \draw[zeroblocked,-] (zi2) -- (cr2);
   \draw[zeroblocked,-,path fading=south] (zi3) -- +(0.00,-1.25);
   \draw[zeroblocked,-,path fading=south] (zi4) -- +(0.00,-1.25);
   \draw[oneblocked,-,path fading=west] (vi) -- +(-0.75,0.75);
   \draw[oneblocked,-,path fading=west] (vi) -- +(-0.75,-0.75);
   \draw[oneblocked,-,path fading=east] (vi1) -- +(0.75,0.75);
   \draw[oneblocked,-,path fading=east] (vi1) -- +(0.75,-0.75);
   \draw[zeroblocked,-,path fading=north] (cr1) -- +(.5,1);
   \draw[zeroblocked,-,path fading=north] (cr1) -- +(-.5,1);
   \draw[zeroblocked,-,path fading=north] (cr2) -- +(-.5,1);
   \draw[zeroblocked,-,path fading=north] (cr2) -- +(.5,1);
   \draw[zeroblocked,-,path fading=north] (cr3) -- +(-.5,1);
   \draw[zeroblocked,-,path fading=north] (cr3) -- +(.5,1);
   \draw[zeroblocked,-,path fading=north] (cr4) -- +(-.5,1);
   \draw[zeroblocked,-,path fading=north] (cr4) -- +(.5,1);
   \path (0,-4) node[draw=none] {$\ldots$};
   \path (7.8,-4) node[draw=none] {$\ldots$};
   \path (1.8,0.69) node[draw=none] {$\mathcal{E}_1$};
   \path (3.3,0.69) node[draw=none] {$\mathcal{E}_2$};
   \path (1.3,-2.5) node[draw=none] {$\mathcal{E}_3$};
  \end{tikzpicture}
  \caption{Schematic representation of the gadget for variable $x_i$, which appears negated in clause~$C_r$ and positive in clause $C_{r+2}$ among others. \label{figure:connectivity_variables}}
 \end{figure}
 
  We call an $(s^+,s^-)$-path that contains no node from $V_3$ a \emph{variable path} and an $(s^+,s^-)$-path with no node from $V_4$ a \emph{clause path}. 
  An $(s^+,s^-)$-path containing edges of type $\mathcal{E}_2$ and $\mathcal{E}_3$ does not connect $s^+$ with $s^-$ in $[t_1,t_1+1]$ or in $[t_2,t_2+1]$.   
  Therefore, all paths other than variable paths and relevant clause paths are irrelevant for the connectivity of~$s^+$ with~$s^-$.
  
  When maintaining all edges of type $\mathcal{E}_1$ in $[t_1,t_2]$, we have connectivity in $[t_2,t_2+1]$ exactly on all variable paths. 
  Conversely, maintaining all edges of type $\mathcal{E}_1$ in $[t_1+1,t_2+1]$ yields connectivity in $[t_1,t_1+1]$ exactly on all relevant clause paths. 
  On the other hand, any clause path can connect $s^+$ with $s^-$ only in $[t_1,t_1+1]$ and any variable path only in $[t_2,t_2+1]$.
  We now claim that there is a schedule with total connectivity time greater than one if and only if the \THREESAT-instance is a \YES-instance.
  
  Let $S$ be a schedule with total connectivity time greater than one.
  Then there is a variable path $P^v$ with positive connectivity time in $[t_2,t_2+1]$ and a clause path $P^c$ with positive connectivity time in $[t_1,t_1+1]$.
  As the total connectivity time is greater than one, $P^c$ cannot walk through both the positive part ($y$ nodes) and the negative part ($z$ nodes) of the gadget for any variable $x_i$.
  This allows to assume w.l.o.g. that $P^v$ and $P^c$ are disjoint between $s'$ and $s^-$.
  Say $P^v$ and $P^c$ share an edge on the negative part ($z$ nodes) of the gadget for variable $x_i$. 
  Then we can redirect the variable path $P^v$ to the positive part ($y$ nodes) without decreasing the total connectivity time. The same works if they share an edge on the positive part.

  Now set $x_i$ to \SATFALSE if $P^v$ uses the nodes $y_i^1,\ldots,y_i^{k_i}$, that is the upper part of the variable gadget, and to \SATTRUE otherwise.
  With this setting, whenever $P^c$ uses edges of a variable gadget, e.g. the sequence $c_r, z_i^{2q-1}, z_i^{2q}, c_{r+1}$ for some $r,q$, disjointness of $P^v$ and $P^c$ implies that clause $C_r$ is satisfied with the truth assignment of variable $x_i$.
  Since every node pair $c_r, c_{r+1}$ is only connected with paths passing through variable gadgets, and at least one of them belongs to $P^c$ we conclude that every clause $C_r$ is satisfied.
 
  Consider a satisfying truth assignment.
  We define a schedule that admits a variable path $P^v$ with connectivity in $[t_2,t_2+1]$.
  This path $P^v$ uses the upper part ($y_i$-part) if $x_i$ is set to \SATFALSE\ and the lower part ($z_i$-part) if $x_i$ is set to \SATTRUE.
  That is, we maintain all edges of type $\mathcal{E}_1$ on the upper path ($y_i$-path) of the variable gadget for $x_i$ in $[t_1,t_2]$ if $x_i$ is \SATFALSE\ and in $[t_1+1,t_2+1]$ if $x_i$ is \SATTRUE. 
  Conversely, edges of type $\mathcal{E}_1$ on the lower path ($z_i$-path) of the variable gadget for $x_i$ are maintained in $[t_1,t_2]$ if $x_i$ is \SATTRUE\ and in $[t_1+1,t_2+1]$ if $x_i$ is \SATFALSE.
  This implies for the part of the gadget for $x_i$ that is not used by $P^v$ that the corresponding edges of type $\mathcal{E}_1$ are scheduled to allow connectivity during $[t_1,t_1+1]$.
  These edges can be used in a clause path to connect node $c_r$ with $c_{r+1}$ for some clauses $C_r$ that is satisfied by the truth assignment of $x_i$. 
  Since all clauses are satisfied by some variable $x_i$ there exists a clause path $P^c$ admitting connectivity in $[t_1,t_1+1]$. 
  Therefore, the constructed schedule allows connectivity during both intervals $[t_1,t_1+1]$ and $[t_2,t_2+1]$.
  
  To show the inapproximability of \MINCONNECTIVITY, we reduce \THREESAT to this problem. We construct an instance of \MINCONNECTIVITY exactly the same way as we did above for \MAXCONNECTIVITY and set $t_1=0$, $t_2=1$, and $T=2$. By definition of the jobs, this results in a instance with only unit-sized jobs. As we discussed above, \YES-instances of \THREESAT result in a \MAXCONNECTIVITY instance with an objective value of 2. For $T = 2$, that means we have connectivity at all time points, and therefore an objective value of 0 for \MINCONNECTIVITY. \NO-instances of \THREESAT on the other result in \MAXCONNECTIVITY instance with an objective value of 1 -- for $T = 2$, this results in \MINCONNECTIVITY objective value of 1 as well. Due to the gap between 1 and 0, any approximation algorithm that outputs a solution within a factor of the optimum solution needs to decide \THREESAT.  \qed
\end{proof}

We reuse the construction in the proof of Theorem~\ref{thm:no-pmtnLowerBound1} repeatedly to obtain the following improved lower bound.

\begin{theorem}
	\label{thm:no-pmtnLowerBound2}
	Unless $\P = \NP$, there is no $(c\sqrt[3]{|E|})$-approximation algorithm for non-preemptive \MAXCONNECTIVITY, for some constant $c > 0$.
\end{theorem}
\begin{proof}
We reuse the construction in the proof of Theorem~\ref{thm:no-pmtnLowerBound1} to  construct a network that has maximum connectivity time $n$ if the given \THREESAT instance is a \YES-instance and maximum connectivity time $1$ otherwise.
This implies that there cannot be an $(n-\epsilon)$-approximation algorithm for non-preemptive \MAXCONNECTIVITY, unless $\P = \NP$.
Here, $n$ is again the number of variables in the given \THREESAT instance.
Note that the construction in the proof of Theorem~\ref{thm:no-pmtnLowerBound1} has $\Theta(n)$ maintenance jobs and thus there exists a constant $c_1>0$ such that $|E| \leq c_1 \cdot n$.
In this proof, we will introduce $\Theta(n^2)$ copies of the construction and thus $|E| \leq c_2 \cdot n^3$ for some $c_2>0$, which gives that $n \geq c_3\sqrt[3]{|E|}$ for some $c_3>0$.
This gives the statement.

For the construction, we use $n^2-n$ copies of the \THREESAT-network from the proof of Theorem~\ref{thm:no-pmtnLowerBound1}, where each one uses different $(t_1,t_2)$-combinations with $t_1,t_2 \in \{0,\ldots,n-1 \}$ and $t_1\neq t_2$.
We use these copies as \THREESAT-gates and mutually connect them as depicted in Figure~\ref{figure:sat-gadgets}.
Recall that for one such \THREESAT-network we have the freedom of choosing the intervals $[t_1,t_1+1]$ and $[t_2,t_2+1]$, which are relevant for connectivity. 
This choice now differs for every $\THREESAT$-gate.

\begin{figure}[htb]
 \centering
 \includegraphics[width=\textwidth]{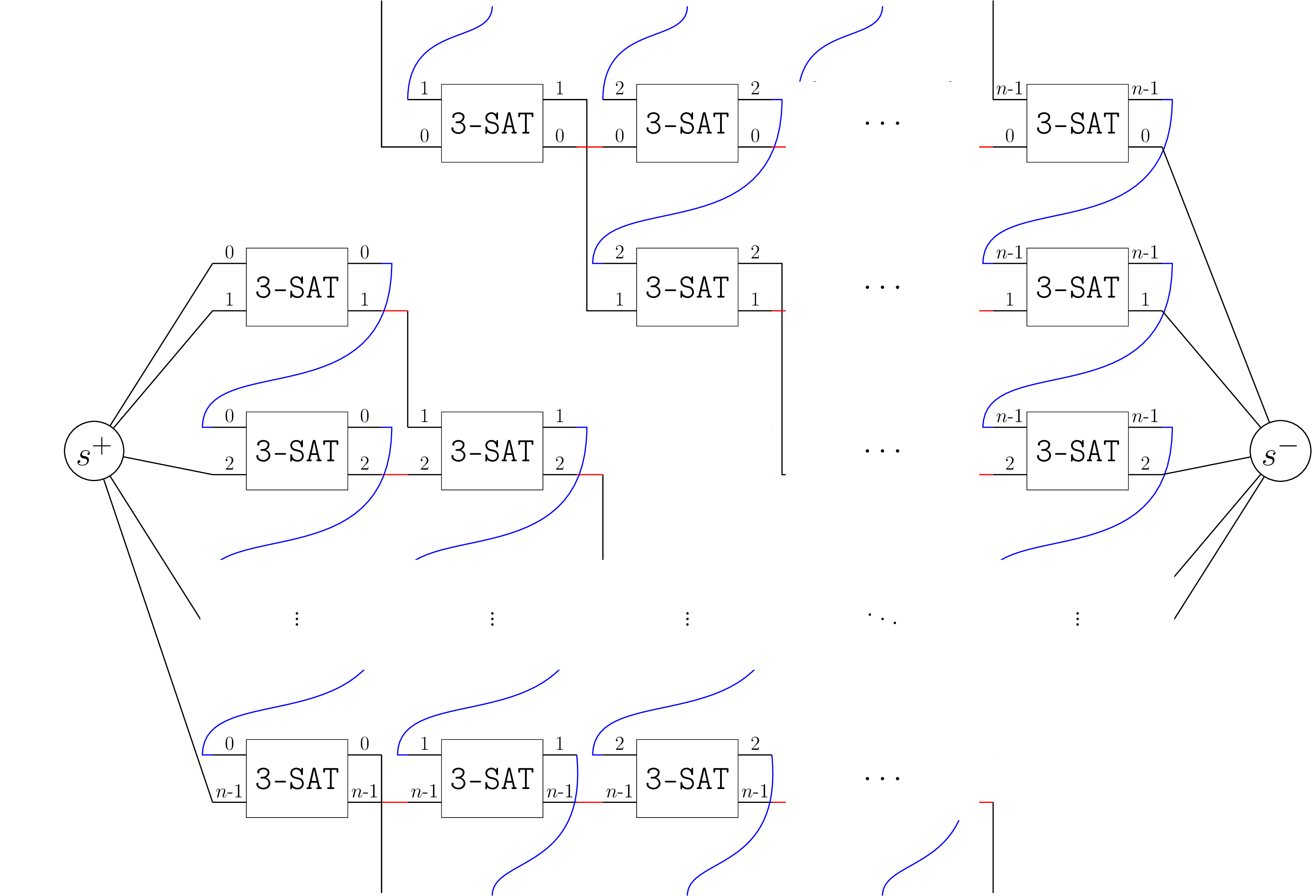}
 \caption{Schematic representation of the network of \THREESAT-gates.\label{figure:sat-gadgets}}
\end{figure}

Think of the construction as an $(n\times n)$-matrix $M$ with an empty diagonal.
Entry $(i,j)$, $i,j\in \{0,\ldots,n-1\}$, in $M$ corresponds to a \THREESAT-gate in that variable paths only exist in time slot $[i,i+1]$ and relevant clause paths exist only in $[j,j+1]$.
This is enforced by the edges of type $\mathcal{E}_2$, which prevent variable paths in $[j,j+1]$, and edges of type $\mathcal{E}_3$, which prevent relevant clause paths in $[i,i+1]$. 
Edges between the $s^+$-copy and $s'$-copy of the \THREESAT-gate$(i,j)$ prevent connectivity outside of $[i,i+1]$ and $[j,j+1]$.
Note that now $\mathcal{E}_1:=\{e\in E: r_e=i, d_e=j+1, p_e=j-i \}$ if $i<j$, and $\mathcal{E}_1:=\{e\in E: r_e=j, d_e=i+1, p_e=i-j \}$ if $i>j$.

The $s^+$-copy of the \THREESAT-gate$(i,j)$ is connected to two paths, where one of them allows connectivity only during $[i,i+1]$ and the other one only during $[j,j+1]$.
The same is done for the $s^-$-copy of the \THREESAT-gate$(i,j)$.
In Figure~\ref{figure:sat-gadgets}, this is illustrated by labels on the paths.
A label $i\in \{0,\ldots,n-1\}$ means, that this path allows connectivity only during $[i,i+1]$.
The upper path connected to a \THREESAT-gate specifies the time slot, where variable paths may exist, and the lower path specifies the time slot, where relevant clause paths may exist.
When following the path with label $k\in \{0,\ldots,n-1\}$, we pass the gadgets in column $j=0,\ldots,k-1$ on the lower path having $j$ on the upper path.
In column $k$, we walk through all gadgets on the upper path and then we proceed with column $j=k+1,\ldots,n-1$ on the lower path having $j$ again on the upper path.
Eventually, we connect the \THREESAT-gate$(n-1,k)$ to the vertex $s^-$.

Note that within \THREESAT-gate$(i,j)$ we have connectivity during $[i,i+1]$ \emph{and} $[j,j+1]$ if and only if the corresponding \THREESAT-instance is a \YES-instance. Also notice that we can assume due to~\cite{BolandNKK15} that all jobs start at integral times, which allows us to ignore schedules with fractional job starting times and therefore fractional connectivity within a time interval $[i,i+1]$. 
Now, if the \THREESAT-instance is a \YES-instance, there is a global schedule such that its restriction to every \THREESAT-gate$(i,j)$ allows connectivity during both intervals. 
Thus for each label $k\in \{0,\ldots,n-1\}$ there exists a path with this label that has connectivity during $[k,k+1]$.
This implies that the maximum connectivity time is $n$.

Conversely, suppose there exists a global schedule with connectivity during $[i,i+1]$ and $[j,j+1]$ for some $i \neq j$.
Then there must exist two paths $P_1, P_2$ from $s^+$ to $s^-$ with two distinct labels $i$ and $j$, each realizing connectivity during one of both intervals. 
By construction they walk through the \THREESAT-gate$(i,j)$. 
This implies by the proof of Theorem~\ref{thm:no-pmtnLowerBound1}, that the global schedule restricted to this gate corresponds to a satisfying truth assignment for the \THREESAT-instance.
That is, the \THREESAT-instance is a \YES-instance. 
With the previous observation, it follows that an optimal schedule has maximum connectivity time of $n$. \qed
\end{proof}

The results above hold for general graph classes, but even for graphs as simple as disjoint paths between $s$ and $t$, the problem remains strongly \NP-hard.

\begin{theorem}
 \label{theorem:connectivity}
 Non-preemptive \MAXCONNECTIVITY is strongly \NP-hard, and non-preemptive \MINCONNECTIVITY is inapproximable even if the given graph consists only of disjoint paths between $s$ and $t$.
\end{theorem}

\begin{proof}

We proof this result by reduction from the strongly \NP-complete \THREESAT problem.

\introduceproblem{\THREESAT}
{Clauses $C_1,\dots,C_m$ of exactly three variables in $x_1, \dots, x_n$.}
{Is there a truth assignment to the variables in $x_1, \dots, x_n$ that satisfies all clauses?}

We construct a network with $2n$ paths from $s^+$ to $s^-$, two for each variable of the \THREESAT{} instance.
Let $P_i$ and~$\bar{P}_i$ denote the two paths for variable~$x_i$.
We will introduce several maintenance jobs for each path, understanding that each new job is associated with a different edge of the path.
Since the ordering of these edges does not matter, we will directly associate each job with a path without explicitly specifying the respective edge of the job.
The network will allow a schedule that maintains connectivity at all times if and only if the \THREESAT{} instance is satisfiable. 

For convenience, assume that $n \geq m$, otherwise we introduce additional dummy variables.
We define a time horizon~$T = 8n$ that we subdivide into five intervals $A = [0, 2n), B = [2n, 3n), C = [3n, 5n), D = [5n, 6n), E = [6n, 8n]$. We will use these intervals now when defining jobs.

\paragraph{Jobs representing variables.}
For each variable~$x_i$, we define a job each on paths~$P_i$ and~$\bar{P}_i$ with the time window~$[0,T]$ and processing time~$3n$.
We will ensure that neither job is scheduled to cover the time interval $C$ entirely in any feasible schedule for the connectivity problem.
This implies that a variable job either covers $B$ or $D$ without intersecting the other.
The job on~$P_i$ (resp.~$\bar{P}_i$) covering $B$ will correspond to the literal~$x_i$ (resp.~$\bar{x}_i$) being set to \SATTRUE.
We will of course ensure that not both literals can be set to \SATTRUE{} simultaneously, but we will allow both to be \SATFALSE, which simply means that the truth assignment remains satisfying, no matter how the variable is set.

\paragraph{Jobs needed to translate schedules into variable assignments.}
In the following, we introduce blocking jobs that all have a time window of unit length and unit processing time.
In this way, introducing a blocking job at time~$t$ simply renders the corresponding path unusable during the time interval~$[t,t+1)$.
To ensure that the variable jobs for variable~$x_i$ do not cover~$C$ completely, we add a blocking job at time~$t_i = 3n + 2(i-1)$ to all paths except~$P_i$ and a blocking job at time~$t'_i = 3n + 2(i-1) + 1$ to all paths except~$\bar{P}_i$.
The first job forces the variable job for the literal~$x_i$ not to cover~$C$ completely, since otherwise connectedness is interrupted during the time interval~$[t_i, t'_i)$.
The second blocking job accomplishes the same for the literal~$\bar{x}_i$.
Note that the blocking jobs for each literal occupy a unique part of the time window~$C$.

\paragraph{Jobs preventing variables from being 0 and 1 at the same time.}
In order to force at most one literal of each variable~$x_i$ to be set to \SATTRUE, we introduce a blocking job at time~$t''_i = 2n + (i-1)$ on all paths except $P_i$ and $\bar{P}_i$.
These blocking jobs ensure that either path $P_i$ or $\bar{P}_i$ must be free during time~$[t''_i, t''_i + 1)$, which means not both variable jobs may be scheduled to cover~$B$ (recall each variable job either covers~$B$ or~$D$ without intersecting the other).
Again, the blocking jobs for each variable occupy a unique part of the time window~$B$.

\paragraph{Jobs enforcing that at least one literal of each clause is true.}
For each clause~$C_j$ we introduce a blocking job at time~$5n + j$ on each path except the three paths that correspond to literals in~$C_j$. Figure~\ref{fig:parallelpathshard} shows this construction for variable $x_i$ and paths $P_i,\bar{P_i}$.
	
 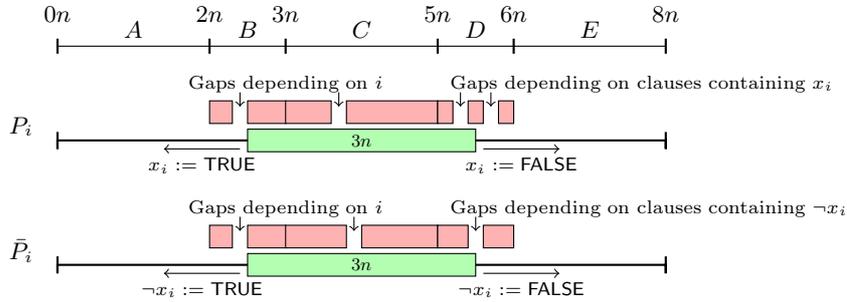
\begin{figure}[hbtp]
  \centering
  \begin{tikzpicture}
   \draw[-] (0,-1) -- (8,-1);
   \foreach \coo in {0,2,3,5,6,8} {
    \draw[thick] (\coo,-1.1) -- (\coo,-.9) node[above=3pt] {$\coo n$};
   }
	 \node[anchor=south] (text) at (1,-1) {$A$};
	 \node[anchor=south] (text) at (2.5,-1) {$B$};
	 \node[anchor=south] (text) at (4,-1) {$C$};
	 \node[anchor=south] (text) at (5.5,-1) {$D$};
	 \node[anchor=south] (text) at (7,-1) {$E$};
   \begin{scope}[yscale=0.75,yshift=-5mm] 
    \node[anchor=south] (text) at (3.0,-1.8) {\scriptsize Gaps depending on $i$}; \draw[->] (2.4,-1.65) -- (2.4,-1.9); \draw[->] (3.7,-1.65) -- (3.7,-1.9);
    \node[anchor=south west] (text) at (5.05,-1.8) {\scriptsize Gaps depending on clauses containing $x_i$}; \draw[->] (5.3,-1.65) -- (5.3,-1.9); \draw[->] (5.7,-1.65) -- (5.7,-1.9);
    \draw[fill=red!30] (2.0,-1.8) rectangle (2.3,-2.2); \draw[fill=red!30] (2.5,-1.8) rectangle (3.0,-2.2);     
    \draw[fill=red!30] (3.0,-1.8) rectangle (3.6,-2.2); \draw[fill=red!30] (3.8,-1.8) rectangle (5.0,-2.2); 
    \draw[fill=red!30] (5.0,-1.8) rectangle (5.2,-2.2); \draw[fill=red!30] (5.4,-1.8) rectangle (5.6,-2.2); \draw[fill=red!30] (5.8,-1.8) rectangle (6.0,-2.2);  
    \draw[fill=green!30] (2.5,-2.3) rectangle (5.5,-2.7); \node[anchor=base,scale=0.75] (text) at (4.0,-2.6) {$3n$}; 
    \draw[thick] (0,-2.35) -- (0,-2.65) --  (0,-2.5) -- (2.5,-2.5); \draw[thick] (8,-2.35) -- (8,-2.65) --  (8,-2.5) -- (5.5,-2.5);
    \draw[<-] (1.4,-2.65) -- node[auto,swap] {\scriptsize $x_i := \SATTRUE$} (2.4,-2.65);
    \draw[->] (5.6,-2.65) -- node[auto,swap] {\scriptsize $x_i := \SATFALSE$} (6.6,-2.65);
    \node[anchor=east] (text) at (-0.2,-2.25) {$P_i$};
    \begin{scope}[yshift=-22mm]
    \node[anchor=south] (text) at (3.0,-1.8) {\scriptsize Gaps depending on $i$}; \draw[->] (2.4,-1.65) -- (2.4,-1.9); \draw[->] (3.9,-1.65) -- (3.9,-1.9);
    \node[anchor=south west] (text) at (5.05,-1.8) {\scriptsize Gaps depending on clauses containing $\neg x_i$}; \draw[->] (5.5,-1.65) -- (5.5,-1.9); 
    \draw[fill=red!30] (2.0,-1.8) rectangle (2.3,-2.2); \draw[fill=red!30] (2.5,-1.8) rectangle (3.0,-2.2);     
    \draw[fill=red!30] (3.0,-1.8) rectangle (3.8,-2.2); \draw[fill=red!30] (4.0,-1.8) rectangle (5.0,-2.2); 
    \draw[fill=red!30] (5.0,-1.8) rectangle (5.4,-2.2); \draw[fill=red!30] (5.6,-1.8) rectangle (6.0,-2.2); 
    \draw[fill=green!30] (2.5,-2.3) rectangle (5.5,-2.7); \node[anchor=base,scale=0.75] (text) at (4.0,-2.6) {$3n$}; 
    \draw[thick] (0,-2.35) -- (0,-2.65) --  (0,-2.5) -- (2.5,-2.5); \draw[thick] (8,-2.35) -- (8,-2.65) --  (8,-2.5) -- (5.5,-2.5);
    \draw[<-] (1.4,-2.65) -- node[auto,swap] {\scriptsize $\neg x_i := \SATTRUE$} (2.4,-2.65);
    \draw[->] (5.6,-2.65) -- node[auto,swap] {\scriptsize $\neg x_i := \SATFALSE$} (6.6,-2.65);
    \node[anchor=east] (text) at (-0.2,-2.25) {$\bar{P_i}$};
    \end{scope}
   \end{scope}
  \end{tikzpicture}
  \caption{The paths $P_i,\bar{P_i}$ for variable $x_i$. The axis marks the times from $0$ to $8n$.}\label{fig:parallelpathshard}
 \end{figure}		

These blocking jobs force that at least one of the literals of the clause has to be set to \SATTRUE, i.e., be scheduled to overlap~$B$ instead of~$D$, otherwise connectivity is interrupted during time~$[5n + j, 5n + j + 1)$.
Note again that the blocking jobs for each clause occupy a unique part of the time window~$D$.

It is now easy to verify that each satisfying truth assignment leads to a feasible schedule without disconnectedness for the connectivity problem and vice versa.

We can use this instance construction for both \MAXCONNECTIVITY and \MINCONNECTIVITY. On the one hand, we have that \YES-instances of \THREESAT result in instances with a \MAXCONNECTIVITY objective value of $T$ and a \MINCONNECTIVITY objective value of $0$, and on the other hand we have that \NO-instances of \THREESAT result in instances with a \MAXCONNECTIVITY objective value $< T$ and a \MINCONNECTIVITY objective value $> 0$. This gives us the strong \NP-hardness for \MAXCONNECTIVITY; for \MINCONNECTIVITY, we get the inapproximability result since the optimal objective value is 0 here, similar to Theorem~\ref{thm:no-pmtnLowerBound1}.\qed
\end{proof}

We give an algorithm that computes an $(\ell+1)$-approximation for non-preemptive \MAXCONNECTIVITY, where $\ell \leq |E|$ is the number of different time points~$d_e-p_e, e\in E$. 
The basic idea is that we consider a set of~$\ell+1$ feasible maintenance schedules, whose total time of connectivity upper bounds the maximum total connectivity time of a single schedule. 
Then the schedule with maximum connectivity time among our set of $\ell+1$ schedules is an~$(\ell+1)$-approximation.

The schedules we consider start every job either immediately at its release date, or at the latest possible time. In the latter case it finishes exactly at the deadline. 
More precisely, for a fixed time point~$t$, we start the maintenance of all edges $e\in E$ with $d_e-p_e \geq t$ at their latest possible start time~$d_e-p_e$. 
All other edges start maintenance at their release date~$r_e$.
This yields at most~$\ell+1 \leq |E|+1$ different schedules $S_t$, as for increasing~$t$, each time point where~$d_e-p_e$ is passed for some edge~$e$ defines a new schedule. 
Algorithm~\ref{alg:nonpreemptive_connectivity_approximation} formally describes this procedure, where $E(t):=\{e\in E: e \text{ is not maintained at } t\}$.

\begin{algorithm}
	\begin{algorithmic}[1]
		\STATE Let~$t_1< \dots< t_\ell$ be all different time points~$d_e-p_e, e\in E$, $t_0=0$ and $t_{\ell+1}=T$.
		\STATE Let~$S_i$ be the schedule, where all edges $e$ with~$d_e-p_e<t_i$ start maintenance at~$r_e$ and all other edges at~$d_e-p_e$, $i =1,\ldots, \ell+1$.
		\STATE For each $S_i$, initialize total connectivity time~$c(t_i) \leftarrow 0, i =1,\ldots, \ell+1$.
		\FOR {$i=1$ to $\ell+1$}
			\STATE Partition the interval~$[t_{i-1}, t_i]$ into subintervals such that each time point~$r_e, r_e+p_e, d_e$, $e\in E$, in this interval defines a subinterval bound.
			\FORALL {subintervals~$[a, b]\subseteq [t_{i-1}, t_i]$}
				\IF{$(V,E(1/2\cdot(a+b)))$ contains an~$(s^+,s^-)$-path for $S_i$}
					\STATE Increase~$c(t_i)$ by~$b-a$.
				\ENDIF
			\ENDFOR
		\ENDFOR
		\RETURN Schedule~$S_i$ for which~$c(t_i), i=1,\ldots,\ell+1,$ is maximized.
	\end{algorithmic}
	\caption{Approx. Algorithm for Non-preemptive \MAXCONNECTIVITY}
	\label{alg:nonpreemptive_connectivity_approximation}
\end{algorithm}

Algorithm~\ref{alg:nonpreemptive_connectivity_approximation} considers finitely many intervals, as all (sub-)interval bounds are defined by a time point~$r_e, r_e+p_e, d_e-p_e$ or~$d_e$ of some~$e\in E$. 
As we can check the network for $(s^+,s^-)$-connectivity in polynomial time, and the algorithm does this for each (sub-)interval, Algorithm~\ref{alg:nonpreemptive_connectivity_approximation} runs in polynomial time.

\begin{theorem}
  \label{thm:napprox}
  Algorithm~\ref{alg:nonpreemptive_connectivity_approximation} is an~$(\ell+1)$-approximation algorithm for non-pre\-emptive \MAXCONNECTIVITY on general graphs, with $\ell\leq |E|$ being the number of different time points~$d_e-p_e, e\in E$.
\end{theorem}

\begin{proof}
 By construction, all schedules~$S_i, i=1,\ldots,\ell+1,$  are feasible and the solution returned has a connectivity time of $\max_{i=1,\dots,\ell+1} c(t_i)$, with $c(t_i)$ being the connectivity time of schedule $S_i$.

 The schedule $S_i$, $i=1,\dots,\ell+1$ is choosen in such a way that the connected time in the interval $[t_{i-1},t_i]$ is maximized. To see this, we need to consider two types of jobs. First, all jobs on edges $e\in E$ with $d_e-p_e \geq t_i$ can be scheduled outside of $[t_{i-1},t_i]$, which is definitely a correct choice in order to maximize the connectivity time in $[t_{i-1},t_i]$. Second, for all edges $e\in E$ with $d_e-p_e < t_i$, we know due to the definition of $t_{i-1}$ that $r_e \leq d_e-p_e \leq t_{i-1}$. Thus, scheduling these jobs at $r_e$ guarantees the least reduction in connectivity time in $[t_{i-1},t_i]$. More precisely, this scheduling disrupts connectivity in the interval $[t_{i-1},r_e+p_e]$ if $t_{i-1} \leq r_e+p_e$, and otherwise not at all. However, all other feasible schedulings must also disrupt connectivity in this interval -- scheduling the job earlier than $r_e$ is not possible, and neither is scheduling the job later than $d_e-p_e \leq t_{i-1}$. Thus, schedule $S_i$ has the maximal connectivity time in $[t_{i-1},t_i]$.
 
 Since the intervals $[t_{i-1},t_i]$, $i=1,\dots,\ell+1$ partition the complete time window $[0,T]$, this allows us to bound the value of the optimal solution $\text{OPT}$ by
 \begin{equation}
  \text{OPT} \leq \sum_{i=1}^{\ell+1} c(t_i) \leq (\ell+1)\max_{i=1,\dots,\ell+1} c(t_i) = (\ell+1)\text{ALG}
\end{equation}	
 with ALG being the value of a solution returned by Algorithm~\ref{alg:nonpreemptive_connectivity_approximation}. This gives us an approximation guarantee of $\ell+1$ and completes our proof.\qed
\end{proof}
\section{Power of Preemption\label{sec:pop}}

We first focus on \MINCONNECTIVITY on a path and analyze how much we can gain by allowing preemption.
First, we show that there is an algorithm that computes a non-preemptive schedule whose value is bounded by $O(\log |E|)$ times the value of an optimal preemptive schedule. 
Second, we argue that one cannot gain more than a factor of $\Omega(\log |E|)$ by allowing preemption.
\begin{theorem}\label{thm:powerofpreemption}
The power of preemption is $\Theta(\log |E|)$ for \MINCONNECTIVITY on a path.	
\end{theorem}

\begin{proof}
Observe that if at least one edge of a path is maintained at time $t$, then the whole path is disconnected at $t$.
We give an algorithm for \MINCONNECTIVITY on a path that constructs a non-preemptive schedule with cost at most $O(\log |E|)$ times the cost of an optimal preemptive schedule.

We first compute an optimal preemptive schedule. This can be done in polynomial time by Theorem~\ref{thm:pmtn}. 
Let $x_t$ be a variable that is 1 if there exists a job $j$ that is processed at time~$t$ and~0 otherwise.
We shall refer to $x$ also as the \emph{maintenance profile}.
Furthermore, let $a:=\int_0^T x_t \ \mathrm{d}t$ be the active time, i.e., the total time of maintenance.
Then we apply the following \emph{splitting procedure}. 
We compute the time point $\bar{t}$ where half of the maintenance is done, i.e., $\int_0^{\bar{t}} x_t \ \mathrm{d}t = a/2$.
Let $E(t) := \{e\in E \mid r_e \leq t \wedge d_e \geq t \}$ and $p_{\max} := \max_{e \in E(t)} p_e$.
We reserve the interval $\left[ \bar{t} - p_{\max},\bar{t} + p_{\max} \right]$ for the maintenance of the jobs in $E(\bar{t})$, although we might not need the whole interval.
We schedule each job in $E(\bar{t})$ around $\bar{t}$ so that the processing time before and after $\bar{t}$ is the same. 
If the release date (deadline) of a jobs does not allow this, then we start (complete) the job at its release date (deadline).
Then we mark the jobs in $E(\bar{t})$ as scheduled and delete them from the preemptive schedule.

\begin{figure}[hbtp]
 \centering
 \begin{tikzpicture}
	\begin{scope}[yscale=0.75,xscale=0.75,yshift=1mm] 
	 \node[anchor=east] (text) at (-0.25,-2.5) {Initial};
	 \node[anchor=east] (text) at (-0.25,-3.0) {Recursion $1$};
	 \node[anchor=east] (text) at (-0.25,-3.5) {Recursion $2$};
	 \draw[fill=orange!30] (4,-2.3) rectangle (8,-2.7); 
	 \draw[thick,dashed] (0,-2.5) -- (4,-2.5); \draw[thick,dashed] (12,-2.5) -- (8,-2.5); \draw[thick] (6,-2.3) -- (6,-2.7);
	 \node[anchor=north] (text) at (6,-2.7) {$\bar{t}$};
	 \draw[fill=orange!30] (1,-2.8) rectangle (3,-3.2);
	 \draw[thick,dashed] (0,-3.0) -- (1,-3.0); \draw[thick,dashed] (4,-3.0) -- (3,-3.0); \draw[thick] (2,-2.8) -- (2,-3.2);
	 \draw[fill=orange!30] (9,-2.8) rectangle (11,-3.2);
	 \draw[thick,dashed] (8,-3.0) -- (9,-3.0); \draw[thick,dashed] (12,-3.0) -- (11,-3.0); \draw[thick] (10,-2.8) -- (10,-3.2); 
	 \draw[fill=orange!30] (0.25,-3.3) rectangle (0.75,-3.7); 
	 \draw[thick,dashed] (0,-3.5) -- (0.25,-3.5); \draw[thick,dashed] (1,-3.5) -- (0.75,-3.5); \draw[thick] (0.5,-3.3) -- (0.5,-3.7);
	 \draw[fill=orange!30] (3.25,-3.3) rectangle (3.75,-3.7); 
	 \draw[thick,dashed] (3,-3.5) -- (3.25,-3.5); \draw[thick,dashed] (4,-3.5) -- (3.75,-3.5); \draw[thick] (3.5,-3.3) -- (3.5,-3.7);
	 \draw[fill=orange!30] (8.25,-3.3) rectangle (8.75,-3.7); 
	 \draw[thick,dashed] (8,-3.5) -- (8.25,-3.5); \draw[thick,dashed] (9,-3.5) -- (8.75,-3.5); \draw[thick] (8.5,-3.3) -- (8.5,-3.7);
	 \draw[fill=orange!30] (11.25,-3.3) rectangle (11.75,-3.7); 
	 \draw[thick,dashed] (11,-3.5) -- (11.25,-3.5); \draw[thick,dashed] (12,-3.5) -- (11.75,-3.5); \draw[thick] (11.5,-3.3) -- (11.5,-3.7);
	\end{scope}
	\end{tikzpicture}
	\caption{A sketch of the splitting procedure and the reserved intervals.}\label{fig:popupper}
\end{figure}
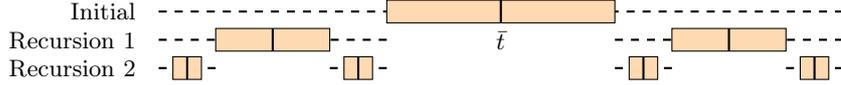	

This splitting procedure splits the whole problem into two separate instances $E_1 := \{e \in E \mid d_e < \bar{t} \}$ and $E_2 := \{e \in E \mid r_e > \bar{t} \}$.
Note that in each of these sub-instances the total active time in the preemptive schedule is at most $a/2$. 
We apply the splitting procedure to both sub-instances and follow the recursive structure of the splitting procedure until all jobs are scheduled.\qed
\end{proof}

\begin{lemma}\label{theorem:benefit-of-pmtn}
For \MINCONNECTIVITY on a path, the given algorithm constructs a non-preemptive schedule with cost $O(\log |E|)$ times the cost of an optimal preemptive schedule.
\end{lemma} 
\begin{proof}
The progression of the algorithm can be described by a binary tree in which a node corresponds to a partial schedule generated by the splitting procedure for a subset of the job and edge set $E$.
The root node corresponds to the partial schedule for $E(\bar{t})$ and the (possibly) two children of the root correspond to the partial schedules generated by the splitting procedure for the two subproblems with initial job sets $E_1$ and $E_2$. 
We can cut a branch if the initial set of jobs is empty in the corresponding subproblem.
We associate with every node $v$ of this tree $B$ two values $(s_v,a_v)$ where $s_v$ is the number of scheduled jobs in the subproblem corresponding to $v$ and $a_v$ is the amount of maintenance time spent for the scheduled jobs. 

The binary tree $B$ has the following properties.
First, $s_v\geq 1$ holds for all $v \in B$, because the preemptive schedule processes some job at the midpoint $\bar{t}_v$ which means that there must be a job $e\in E$ with $r_e \leq \bar{t}_v \wedge d_e \geq \bar{t}_v$. 
This observation implies that the tree $B$ can have at most $|E|$ nodes and since we want to bound the worst total cost we can assume w.l.o.g. that $B$ has exactly $|E|$ nodes. 
Second, $\sum_{v\in B} a_v = \int_0^T y_t \ \mathrm{d}t$ where $y_t$ is the maintenance profile of the non-preemptive solution.

The cost $a_v$ of the root node (level-0 node) is bounded by $2p_{\max} \leq 2a$. The cost of each level-1 node is bounded by $2\cdot a/2 = a$, so the total cost on level 1 is also at most $2a$. 
It is easy to verify that this is invariant, i.e., the total cost at level $i$ is at most $2a$ for all $i \geq 0$, since the worst node cost $a_v$ halves from level $i$ to level $i+1$, but the number of nodes doubles in the worst case.  
We obtain the worst total cost when $B$ is a complete balanced binary tree. This tree has at most $O(\log |E|)$ levels and therefore the worst total cost is $a \cdot O(\log |E|)$. The total cost of the preemptive schedule is $a$. \qed
\end{proof}

\noindent
We now provide a matching lower bound for the power of preemption on a path.

\begin{lemma}\label{lemma:price_of_nonpmtn} 
  The power of non-preemption is $\Omega(\log |E|)$ for \MINCONNECTIVITY on a path.
\end{lemma}
\begin{proof}
We construct a path with $|E|$ edges and divide the $|E|$ jobs into $\ell$ levels such that level~$i$ contains exactly $i$ jobs for $1\leq i\leq \ell$. 
Hence, we have $|E|=\ell(\ell+1)/2$ jobs. Let $P$ be a sufficiently large integer such that all of the following numbers are integers.
Let the $j$th job of level $i$ have release date $(j-1)P/i$, deadline $(j/i)P$, and processing time $P/i$, where $1\leq j\leq i$. 
Note that now no job has flexibility within its time window, and thus the value of the resulting schedule is~$P$.  

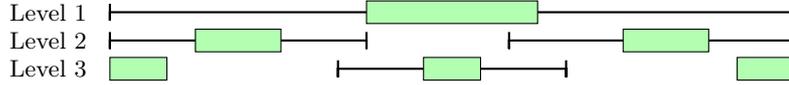
\begin{figure}[hbtp]
 \centering
 \begin{tikzpicture}
	\begin{scope}[yscale=0.75,yshift=1mm,xscale=0.75] 
	 \node[anchor=east] (text) at (-0.25,-2.5) {Level $1$};
	 \node[anchor=east] (text) at (-0.25,-3.0) {Level $2$};
	 \node[anchor=east] (text) at (-0.25,-3.5) {Level $3$};
	 \draw[fill=green!30] (4.5,-2.3) rectangle (7.5,-2.7); 
	 \draw[thick] (0,-2.35) -- (0,-2.65) --  (0,-2.5) -- (4.5,-2.5); \draw[thick] (12,-2.35) -- (12,-2.65) --  (12,-2.5) -- (7.5,-2.5);
	 \draw[fill=green!30] (1.5,-2.8) rectangle (3.0,-3.2);
	 \draw[thick] (0,-2.85) -- (0,-3.15) --  (0,-3.0) -- (1.5,-3.0); \draw[thick] (4.5,-2.85) -- (4.5,-3.15) --  (4.5,-3.0) -- (3.0,-3.0);
	 \draw[fill=green!30] (9.0,-2.8) rectangle (10.5,-3.2);
	 \draw[thick] (7,-2.85) -- (7,-3.15) --  (7,-3.0) -- (9.0,-3.0); \draw[thick] (12,-2.85) -- (12,-3.15) --  (12,-3.0) -- (10.5,-3.0);
	 \draw[fill=green!30] (0.00,-3.3) rectangle (1.00,-3.7); 
	 \draw[fill=green!30] (5.5,-3.3) rectangle (6.5,-3.7); 
	 \draw[thick] (4,-3.35) -- (4,-3.65) --  (4,-3.5) -- (5.5,-3.5); \draw[thick] (8,-3.35) -- (8,-3.65) --  (8,-3.5) -- (6.5,-3.5);		
	 \draw[fill=green!30] (11,-3.3) rectangle (12,-3.7); 	
	\end{scope}
	\end{tikzpicture}
	\caption{A rough sketch of the instance for 3 levels.}\label{fig:poplower}
\end{figure}	

We now modify the instance as follows. 
At every time point $t$ where at least one job has a release date and another job has a deadline, we stretch the time horizon by inserting a gap of size P. 
This stretching at time $t$ can be done by adding a value of $P$ to all time points after the time point $t$, and also adding a value of $P$ to all release dates at time $t$. 
The deadlines up to time $t$ remain the same. 
Observe that the value of the optimal preemptive schedule is still $P$, because when introducing the gaps we can move the initial schedule accordingly such that we do not maintain any job within the gaps of size $P$. Figure~\ref{fig:poplower} shows a rough sketch of this construction. 

We now consider the optimal non-preemptive schedule. 
The cost of scheduling the only job at level $1$ is $P$. 
In parallel to this job we can schedule at most one job from each other level, without having additional cost. 
This is guaranteed by the introduced gaps. 
At level~$2$ we can fix the remaining job with additional cost $P/2$. 
As before, in parallel to this fixed job, we can schedule at most one job from each level $i$ where $3\leq i\leq \ell$. 
Applying the same argument to the next levels, we notice that for each level $i$ we introduce an additional cost of value $P/i$. 
Thus the total cost is at least $\sum^{\ell}_{i=1} P/i \in \Omega(P\log \ell)$ with $\ell \in \Theta(\sqrt{|E|})$. \qed
\end{proof}

Next, we show that for \MAXCONNECTIVITY, the power of preemption can be unbounded.

\begin{theorem}
 \label{thm:powerofpreemptionmax}
 For non-preemptive \MAXCONNECTIVITY on a path the power of preemption is unbounded.
\end{theorem}

\begin{proof}
	Consider a path of four consecutive edges~$e_1=\{s^+,u\},e_2=\{u,w\},e_3=\{w,v\},e_4=\{v,s^-\}$, each associated with a maintenance job as depicted in Figure~\ref{fig:unbounded_integrality_gap_on_a_path}.
That is, $r_1=r_2=0,d_1=r_3=p_1=p_4=1,p_2=p_3=2,r_4=d_2=3,d_3=d_4=4$.	

\begin{figure}[hbtp]
 \centering
 \begin{tikzpicture}
	\draw[-] (0,-1) node[left] {$t$} -- (1,-1)  
					-- (2,-1) 
					-- (3,-1)  
					-- (4,-1) ;
				\draw[thick] (0,-1.1) -- (0,-.9);
				\foreach \coo in {1,2,3,4}
				{
				  \draw[thick] (\coo,-1.1) -- (\coo,-.9) node[above=3pt] {\coo};
				}
				\begin{scope}[yscale=0.75,yshift=1mm] 
				\draw[fill=red!30] (0,-1.8) rectangle (1,-2.2); \node[anchor=east] (text) at (0,-2) {$e_1$};
				\draw[fill=green!30] (0.5,-2.3) rectangle (2.5,-2.7); \node[anchor=east] (text) at (0,-2.5) {$e_2$}; 
				\draw[thick] (0,-2.35) -- (0,-2.65) --  (0,-2.5) -- (0.5,-2.5); \draw[thick] (3,-2.35) -- (3,-2.65) --  (3,-2.5) -- (2.5,-2.5);
				\draw[fill=green!30] (1.5,-2.8) rectangle (3.5,-3.2); \node[anchor=east] (text) at (0,-3) {$e_3$}; 
				\draw[thick] (1,-2.85) -- (1,-3.15) --  (1,-3.0) -- (1.5,-3.0); \draw[thick] (4,-2.85) -- (4,-3.15) --  (4,-3.0) -- (3.5,-3.0);
				\draw[fill=red!30] (3,-3.3) rectangle (4,-3.7); \node[anchor=east] (text) at (0,-3.5) {$e_4$};
				\end{scope}
	\end{tikzpicture}
	\caption{Example for an unbounded power of preemption.}\label{fig:unbounded_integrality_gap_on_a_path}
\end{figure}
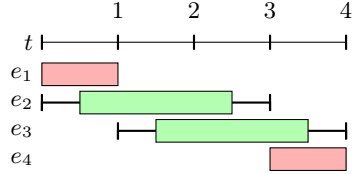	
	There is no non-preemptive schedule that allows connectivity at any point in time, as the maintenance job of edge~$e_i$ blocks edge~$e_i$ in time slot~$[i-1,i]$. 
	On the other hand, when allowing preemptive schedules, we can process the job of edge~$e_2$ in $[0,2]$ and the job of edge~$e_3$ in $[1,2]$ and $[3,4]$.
	Then no maintenance job is scheduled in the time interval~$[2,3]$ and therefore we have connectivity for one unit of time. \qed
\end{proof} 

\section{Mixed Scheduling}\label{sec:mixed}

We know that both the non-preemptive and preemptive \MAXCONNECTIVITY
and \MINCONNECTIVITY on a path are solvable in polynomial time by
Theorem~\ref{thm:pmtn} and \cite[Theorem 9]{KhandekarEtAl15},
respectively. Notice that the parameter $g$ in~\cite{KhandekarEtAl15}
is in our setting $\infty$.
Interestingly, the complexity changes when mixing the two job types -- even on a simple path.

\begin{theorem}
 \label{theorem:mixed}
 \MAXCONNECTIVITY and \MINCONNECTIVITY with preemptive and non-preemptive maintenance jobs is weakly \NP-hard, even on a path.
\end{theorem}

\begin{proof}
We reduce the \NP-hard \PARTITION~problem to \MAXCONNECTIVITY. We will show that there is a gap in the objective value between instances derived from \YES- and \NO-instances of \PARTITION, respectively. This gap is same for \MINCONNECTIVITY, since maximizing the time in which we have connectivity is the same as minimizing the time in which we do not have connectivity.
	
	\introduceproblem{\PARTITION}
	{A set of $n$ natural numbers $A = \{a_1, \dots, a_n\} \subset \N$ with $\sum_{i=1}^n a_i = 2B$ for some $B \in \N$.}
	{Is there a subset $S \subseteq A$ with $\sum_{a \in S}a = B$?}
	
Given an instance of \PARTITION, we create a \MAXCONNECTIVITY instance based on a path consisting of $3n+2$ edges between $s^+$ and $s^-$ with preemptive and non-preemptive maintenance jobs.  
We create three types of job sets denoted as $J_1, J_2$ and $J_3$, where the first two job sets model the binary decision involved in choosing a subset of numbers to form a partition, whereas the third job set performs the summation over the numbers picked for a partition. The high-level idea is depicted in Figure~\ref{fig:HardnessMixedShort}.

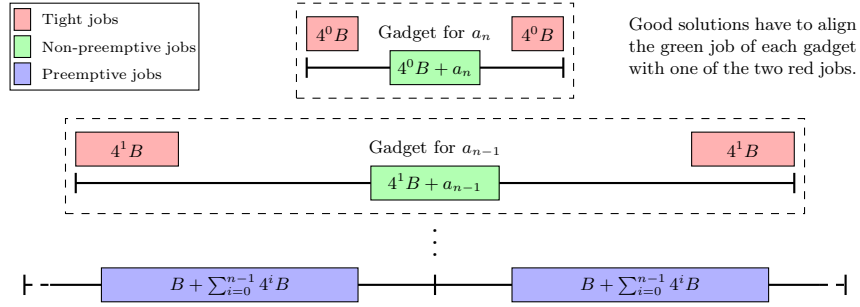
\begin{figure}[hbtp]
 \centering
 \begin{tikzpicture}[scale=0.9,xscale=0.75,every node/.style={scale=0.75}]
  \jobAShort{4cm}{-2.5}{1}{1.5}{4^0B}{4^0B}
  \jobBShort{3.5cm}{-1.5+0.625}{1.75}{4^0B+a_n}{x_n+a_n}{-2.5}{2.5}
	\draw[dashed] (-2.7,3.3) rectangle (2.7,4.7);
	\node (text) at (0,4.25) {\footnotesize Gadget for $a_n$};
	\node (expln) at (6,4.05) {\begin{minipage}{4cm}\footnotesize Good solutions have to align the green job of each gadget with one of the two red jobs.\end{minipage}};
	\begin{scope}[yshift=-3mm]
   \jobAShort{2.6cm}{-7}{2}{5}{4^1B}{4^1B}
   \jobBShort{2.1cm}{-1.25}{2.5}{4^1B+a_{n-1}}{x_{n-1}+a_{n-1}}{-7}{7}   
	 \draw[dashed] (-7.2,1.9) rectangle (7.2,3.3);
	 \node (text) at (0,2.85) {\footnotesize Gadget for $a_{n-1}$};
	\end{scope}
  \node (x1) at (0,1.3) {\scalebox{1.5}{$\vdots$}};
	\begin{scope}[yshift=2.3cm]
	 \jobCShort{-2.0cm}{-6.5}{4}{B+ \sum_{i=0}^{n-1} 4^{i}B}{B+ \sum_{i=1}^n x_i}{-7.5}{0}
	 \jobCShort{-2.0cm}{1.5}{4}{B+ \sum_{i=0}^{n-1} 4^{i}B}{B+ \sum_{i=1}^n x_i}{0}{7.5}
	 \draw[thick] ($(0,-2)+(0,0.1)$) -- ($(0,-2)+(0,0.4)$);
	 \draw[thick,dashed] ($(-8,-2)+(0,0.25)$) -- ($(-7.5,-2)+(0,0.25)$);
	 \draw[thick,dashed] ($(7.5,-2)+(0,0.25)$) -- ($(8,-2)+(0,0.25)$);
   \draw[thick] ($(-8,-1.75)+(0,0.15)$) -- ($(-8,-1.75)+(0,-0.15)$);
   \draw[thick] ($(8,-1.75)+(0,0.15)$) -- ($(8,-1.75)+(0,-0.15)$);	
	\end{scope}
	\begin{scope}[xshift=-3cm,yshift=0.7cm,scale=0.8,every node/.style={scale=0.8}]
	 \draw[draw=black,fill=red!30] (-6.5,4.5) rectangle ($(-6.5,4.5)+(0.4,0.4)$);%
	 \node[anchor=west] (text) at (-6,4.675) {\scriptsize Tight jobs};
	 \draw[draw=black,fill=green!30] (-6.5,4.0) rectangle ($(-6.5,4.0)+(0.4,0.4)$);%
	 \node[anchor=west] (text) at (-6,4.175) {\scriptsize Non-preemptive jobs};
	 \draw[draw=black,fill=blue!30] (-6.5,3.5) rectangle ($(-6.5,3.5)+(0.4,0.4)$);%
	 \node[anchor=west] (text) at (-6,3.675) {\scriptsize Preemptive jobs};
	 \draw (-6.6,5) rectangle (-2.0,3.4);
	\end{scope}
 \end{tikzpicture}
 \caption{Instance created from a \PARTITION instance $a_1,\dots,a_n,B$. The number inside the blocks are the processing times of the jobs.\label{fig:HardnessMixedShort}}
\end{figure}

The job set $J_1:= \{1,2,\ldots,2n-1,2n\}$ contains $2n$ \emph{tight} jobs, i.e., $r_j+p_j = d_j$ for all $j\in J_1$. 
For every element $a_i \in A$ we have two tight jobs $i$ and $2n-(i-1)$ both having processing time $4^{n-i}B =: x_i$.
The release date of a job $j\in \{2,\ldots,n\} \subset J_1$ is $r_j = \sum_{k=1}^{j-1} 2x_k + a_k$ and $r_1 = 0$.
Let $\tau := \sum_{k=1}^{n} 2x_k + a_k$.
For $j\in \{n+1,\ldots,2n\} \subset J_1$ we have $d_j = \tau + \sum_{k=n+1}^{j} 2x_{2n-k+1} + a_{2n-k+1}$.
Note that the tight jobs in $J_1$ are constructed in such a way that everything is symmetric with respect to the time point $\tau$.

The job set $J_2:= \{2n+1,\ldots, 3n\}$ contains $n$ non-preemptive jobs.
Let $j_i := 2n+i$.
For every element $a_i \in A$ we introduce job $j_i$ with processing time $p_{j_i}=x_i+a_i$, release date $r_{j_i}= r_i$, and deadline $d_{j_i} = d_{2n-(i-1)}$.
Again, everything is symmetric with respect to time point $\tau$. 

Finally, the set $J_3 := \{3n+1,3n+2\}$ contains two preemptive jobs, where each of them has processing time $W:=B+ \sum_{i=1}^n x_i$.
Furthermore, we have $r_{3n+1}=0$, $d_{3n+1}=\tau$, $r_{3n+2}=\tau$, $d_{3n+2}=2\tau$.

We now show that there is a feasible schedule for the constructed instance that disconnects the path for at most $2W$ time units if and only if the given \PARTITION instance is a \YES-instance.

Suppose there is a subset $S \subseteq A$ with $\sum_{a \in S}a = B$. 
For each $a_i \in S$, we start the corresponding job $j_i \in J_2$ at its release date and the remaining jobs in $J_2$ corresponding to the elements $a_i \in A\setminus S$ are scheduled such that they complete at their deadline.    
This creates $B+\sum_{i=1}^n x_i$ time slots in both intervals $[0,\tau]$ and $[\tau,2\tau]$ with no connection between $s^+$ and $s^-$.
The jobs $3n+1$ and $3n+2$ can be preempted in $[0,\tau]$ and $[\tau,2\tau]$, respectively, and thus if we align their processing with the chosen maintenance slots, we get a schedule that disconnects $s^+$ and $s^-$ for $2W=2(B+\sum_{i=1}^n x_i)$ time units. 

Conversely, suppose that there is a feasible schedule for the constructed instance that disconnects the path for at most $2W$ time units.
By induction on $i$, we show that every job $j_i=2n+i$ either starts at its release date or it completes at its deadline in such a schedule. 

Consider the base case of $i=1$.
We first observe that w.l.o.g. job $j_1$ either starts at its release date or completes at its deadline or is scheduled somewhere in $[x_1, 2\tau - x_1]$.
Suppose it starts somewhere in $(0,x_1)$ or completes somewhere in $(\tau-x_1,\tau)$.
Then we do not increase the total time where the path is disconnected if we push job $j_1$ completely to the left or completely to the right.
If we schedule job $j_1$ in $[x_1, 2\tau - x_1]$, then the total time where the path is disconnected is at least $3x_1 + a_1 > 2x_1 +x_1$.
We will now show that $x_1 \geq 2(B + \sum_{k=2}^n x_k)$ for $n \geq 2$, which shows that the path is then disconnected for more than $2W$ time units, and thus job $j_1$ cannot be processed in $[x_1, 2\tau - x_1]$.
The inequality is true for $n\geq 2$, since
\ama
	2B + 2 \sum_{k=2}^n x_k &= 2B(1 + \sum_{k=2}^n 4^{n-k})\\ &= 2B(1+ \sum_{k=0}^{n-2} 4^{k})\\ &=2B(1+ 1/3(4^{n-1} -1))\\ &\leq 4^{n-1} B = x_1.
\ema
This finishes the proof for $i=1$.

Suppose, the statement is true for $i=1,\ldots,\ell-1$ with $\ell\in \{2,\ldots,n-1\}$.
As in the base case, we can show that job $j_{\ell}$ either starts at its release date or completes at its deadline or is scheduled somewhere in $[r_{j_{\ell}}+x_{\ell}, d_{j_{\ell}}-x_{\ell}]$. 
If job $j_{\ell}$ is processed in $[r_{j_{\ell}}+x_{\ell}, d_{j_{\ell}}-x_{\ell}]$, then the total time where the path is disconnected is at least \[\sum_{k=1}^{\ell -1} (2x_k + a_k ) + 3x_{\ell} + a_{\ell} > \sum_{k=1}^{\ell}2x_k + x_\ell. \]
Again, we will show that $x_\ell \geq 2(B + \sum_{k=\ell+1}^n x_k)$ for $\ell\in \{2,\ldots,n-1\}$, which shows that the path is then disconnected for more than $2W$ time units, and thus job $j_{\ell}$ cannot be processed in $[r_{j_{\ell}}+x_{\ell}, d_{j_{\ell}}-x_{\ell}]$.
The inequality is true for $\ell\in \{2,\ldots,n-1\}$, since
\ama
	2B + 2 \sum_{k=\ell+1}^n x_k &= 2B(1 + \sum_{k=\ell+1}^n 4^{n-k})\\ &= 2B(1+ \sum_{k=0}^{n-\ell -1} 4^{k})\\ &=2B(1+ 1/3(4^{n-\ell} -1))\\
	&\leq 4^{n-\ell} B = x_\ell.
\ema

For $i=n$, we again use the fact that $j_n$ either starts at its release date or completes at its deadline or is scheduled somewhere in $[r_{j_{n}}+x_{n}, d_{j_{n}}-x_{n}]$.
If the latter case is true, then the total time where the path is disconnected is at least 
\begin{eqnarray*}
	\sum_{k=1}^{n-1}(2x_k +a_k) + 3x_n +a_n &=& \sum_{k=1}^{n}(2x_k +a_k) + x_n \\
	&> & 2( B + \sum_{k=1}^{n}x_k) = 2W.
\end{eqnarray*}

There is a feasible schedule for the constructed instance that disconnects the path for at most $2(B+\sum_{k=1}^n x_k)$ time units.
This means that in both $[0,\tau]$ and $[\tau,2\tau]$ the path is disconnected for exactly $B + \sum_{k=1}^n x_k$ time units.
Consider the set $S:=\{i : j_i \mbox{ starts at its release date} \}$.
We conclude that 
\begin{equation}
 \sum_{k=1}^n x_k + \sum_{k\in S} a_k = \sum_{k=1}^n x_k + \sum_{k\notin S} a_k = \sum_{k=1}^n x_k + B .
\end{equation} \qed
\end{proof}

For \MINCONNECTIVITY, running the optimal preemptive and non-pre\-emp\-ti\-ve algorithms on the respective job sets individually gives a $2$-approximation. 

\begin{theorem}
  \label{theorem:mixedapprox}
  There is a 2-approximation algorithm for \MINCONNECTIVITY on a path with preemptive and non-preemptive maintenance jobs.	 
\end{theorem}

\begin{proof}
Consider an optimal schedule $S^*$ for the mixed instance and let $|S^*|$ be the total time of disconnectivity in $S^*$.
Furthermore, let $S^*_{np}$ (resp. $S^*_{p}$) be the restriction of $S^*$ to only non-preemptive (resp. preemptive) jobs.
Note that the schedule $S^*_{np}$ (resp. $S^*_{p}$) is feasible for the corresponding non-preemptive (resp. preemptive) instance.
We separate the preemptive from the non-preemptive jobs and obtain two separate instances.  
Solving them individually in polynomial time and combining the resulting two solutions $S_{np}$ and $S_{p}$ to a schedule $S$ gives the claimed result, because
$|S| \leq |S_{np}| + |S_{p}| \leq |S^*_{np}| + |S^*_{p}| \leq 2 |S^*|$.
 \qed      	
\end{proof}
\section{Conclusion} 
Combining network flows with scheduling aspects is a very recent field of research. While there are solutions using IP based methods and heuristics, 
exact and approximation algorithms have not been considered extensively. We provide strong hardness results for connectivity problems, which is inherent to all forms of maintenance scheduling, and give algorithms for tractable cases. 
 
In particular, the absence of $c\sqrt[3]{|E|}$-approximation algorithms for some $c > 0$ for general graphs indicates that heuristics and IP-based methods~\cite{BolandKK15,BolandNKK15,BolandKWZ14} are a good way of approaching this problem. An interesting open question is whether the inapproximability results carry over to series-parallel graphs, as the network motivating~\cite{BolandKK15,BolandNKK15,BolandKWZ14} is series-parallel.
Our results on the power of preemption as well as the efficient algorithm for preemptive instances show that allowing preemption is very desirable. Thus, it could be interesting to study models where preemption is allowed, but comes at a cost to make it more realistic.

On a path, our results have implications for minimizing busy time, as we want to minimize the number of times where some edge on the path is maintained. 
Here, an interesting open question is whether the 2-approximation for the mixed case can be improved, e.g. by finding a pseudo-polynomial algorithm, a better approximation ratio, or conversely, to show an inapproximability result for it.

\paragraph{Acknowledgements.} We thank the anonymous reviewers for their helpful comments.

\bibliography{literature}

\begin{thebibliography}{10}
\providecommand{\url}[1]{\texttt{#1}}
\providecommand{\urlprefix}{URL }

\bibitem{Bley13}
Bley, A., Karch, D., D'Andreagiovanni, F.: {WDM} fiber replacement scheduling.
  Electronic Notes in Discrete Mathematics  41,  189--196 (2013),
  \url{http://www.sciencedirect.com/science/article/pii/S1571065313000954}

\bibitem{BolandKK15}
Boland, N., Kalinowski, T., Kaur, S.: Scheduling arc shut downs in a network to
  maximize flow over time with a bounded number of jobs per time period.
  Journal of Combinatorial Optimization pp. 1--21 (2015),
  \url{http://dx.doi.org/10.1007/s10878-015-9910-x}

\bibitem{BolandNKK15}
Boland, N., Kalinowski, T., Kaur, S.: Scheduling network maintenance jobs with
  release dates and deadlines to maximize total flow over time: Bounds and
  solution strategies. Computers \& Operations Research  64,  113--129 (2015),
  \url{http://www.sciencedirect.com/science/article/pii/S0305054815001288}

\bibitem{BolandKWZ14}
Boland, N., Kalinowski, T., Waterer, H., Zheng, L.: Scheduling arc maintenance
  jobs in a network to maximize total flow over time. Discrete Applied
  Mathematics  163,  34--52 (2014),
  \url{http://dx.doi.org/10.1016/j.dam.2012.05.027}

\bibitem{BolandS12}
Boland, N.L., Savelsbergh, M.W.P.: Optimizing the hunter valley coal chain. In:
  Gurnani, H., Mehrotra, A., Ray, S. (eds.) Supply Chain Disruptions: Theory
  and Practice of Managing Risk. pp. 275--302. Springer, London (2012),
  \url{http://dx.doi.org/10.1007/978-0-85729-778-5_10}

\bibitem{CanettiIrani98}
Canetti, R., Irani, S.: Bounding the power of preemption in randomized
  scheduling. SIAM Journal on Computing  27(4),  993--1015 (1998),
  \url{http://dx.doi.org/10.1137/S0097539795283292}

\bibitem{ChangKhullerMukherjee14}
Chang, J., Khuller, S., Mukherjee, K.: {LP} rounding and combinatorial
  algorithms for minimizing active and busy time. In: Blelloch, G.E., Sanders,
  P. (eds.) Proc. of the 26th {SPAA}. pp. 118--127. ACM, New York (2014),
  \url{http://doi.acm.org/10.1145/2612669.2612689}

\bibitem{khuller-mapsp}
Chang, J., Khuller, S., Mukherjee, K.: Active and busy time minimization. In:
  Proc. of the 12th {MAPSP}. pp. 247--249 (2015),
  \url{http://feb.kuleuven.be/mapsp2015/Proceedings%20MAPSP%202015.pdf}

\bibitem{CohenAddadEtAl15}
Cohen-Addad, V., Li, Z., Mathieu, C., Milis, I.: Energy-efficient algorithms
  for non-preemptive speed-scaling. In: Bampis, E., Svensson, O. (eds.) Proc.
  of the 12th WAOA. LNCS, vol. 8952, pp. 107--118. Springer International
  Publishing (2015), \url{http://dx.doi.org/10.1007/978-3-319-18263-6_10}

\bibitem{CorreaSkutellaVerschae12}
Correa, J.R., Skutella, M., Verschae, J.: The power of preemption on unrelated
  machines and applications to scheduling orders. Mathematics of Operations
  Research  37(2),  379--398 (2012),
  \url{http://dx.doi.org/10.1287/moor.1110.0520}

\bibitem{FlamminiEtAl10}
Flammini, M., Monaco, G., Moscardelli, L., Shachnai, H., Shalom, M., Tamir, T.,
  Zaks, S.: Minimizing total busy time in parallel scheduling with application
  to optical networks. Theoretical Computer Science  411(40--42),  3553--3562
  (2010),
  \url{http://www.sciencedirect.com/science/article/pii/S0304397510002926}

\bibitem{Ha92}
Ha, S.: Compile-time scheduling of dataflow program graphs with dynamic
  constructs. Ph.D. thesis, University of California, Berkeley (1992),
  \url{http://www.eecs.berkeley.edu/Pubs/TechRpts/1992/ERL-92-43.pdf}

\bibitem{KalinowskiMatsypuraSavelsbergh15}
Kalinowski, T., Matsypura, D., Savelsbergh, M.W.: Incremental network design
  with maximum flows. European Journal of Oper. Res.  242(1),  51--62 (2015),
  \url{http://www.sciencedirect.com/science/article/pii/S0377221714008078}

\bibitem{KhandekarEtAl15}
Khandekar, R., Schieber, B., Shachnai, H., Tamir, T.: Real-time scheduling to
  minimize machine busy times. Journal of Scheduling  18(6),  561--573 (2015),
  \url{http://dx.doi.org/10.1007/s10951-014-0411-z}

\bibitem{PathDecomp}
Korte, B., Vygen, J.: Combinatorial Optimization: Theory and Algorithms.
  Springer Publishing Company, Incorporated, 4th edn. (2007)

\bibitem{MertziosEtAl12}
Mertzios, G.B., Shalom, M., Voloshin, A., Wong, P.W.H., Zaks, S.: Optimizing
  busy time on parallel machines. In: Proc. of the 26th IPDPS. pp. 238--248.
  IEEE (2012),
  \url{http://ieeexplore.ieee.org/xpl/articleDetails.jsp?arnumber=6267839}

\bibitem{NurreEtAl12}
Nurre, S.G., Cavdaroglu, B., Mitchell, J.E., Sharkey, T.C., Wallace, W.A.:
  Restoring infrastructure systems: An integrated network design and scheduling
  {(INDS)} problem. European Journal of Operational Research  223(3),  794--806
  (2012),
  \url{http://www.sciencedirect.com/science/article/pii/S0377221712005310}

\bibitem{ParsonsSevcik95}
Parsons, E.W., Sevcik, K.C.: Multiprocessor scheduling for high-variability
  service time distributions. In: Feitelson, D.G., Rudolph, L. (eds.) Proc. of
  the JSSPP. LNCS, vol. 949, pp. 127--145. Springer Berlin Heidelberg (1995),
  \url{http://link.springer.com/chapter/10.1007%2F3-540-60153-8_26}

\bibitem{SchulzSkut02}
Schulz, A.S., Skutella, M.: Scheduling unrelated machines by randomized
  rounding. SIAM Journal on Discrete Mathematics  15(4),  450--469 (2002),
  \url{http://dx.doi.org/10.1137/S0895480199357078}

\bibitem{SoperS14}
Soper, A.J., Strusevich, V.A.: Power of preemption on uniform parallel
  machines. In: Proc. of the 17th APPROX. LIPIcs, vol.~28, pp. 392--402.
  Schloss Dagstuhl--Leibniz-Zentrum fuer Informatik, Dagstuhl, Germany (2014),
  \url{http://drops.dagstuhl.de/opus/volltexte/2014/4711}

\end{thebibliography}

\end{document}